\documentclass{article}

 \usepackage[main, final]{neurips_2026}

\usepackage[utf8]{inputenc} 
\usepackage[T1]{fontenc}    
\usepackage{hyperref}       
\usepackage{url}            
\usepackage{booktabs}       
\usepackage{amsfonts}       
\usepackage{nicefrac}       
\usepackage{microtype}      
\usepackage{xcolor}         

\usepackage{microtype}
\usepackage{graphicx}
\usepackage{booktabs}
\usepackage{multirow}
\usepackage[table,xcdraw]{xcolor}
\usepackage{subcaption}
\usepackage{booktabs} 

\usepackage{wrapfig}
\usepackage{etoc} 

\usepackage{hyperref}
\usepackage{rotating} 
\newcommand{\rotcell}[1]{%
  \rotatebox[origin=c]{90}{\parbox{1.5cm}{\centering #1}}%
}

\usepackage{algorithm}
\usepackage{algorithmic}

\usepackage{amsmath}
\usepackage{amssymb}
\usepackage{mathtools}
\usepackage{amsthm}
\usepackage{tcolorbox}

\DeclareMathOperator*{\argmax}{arg\,max}

\newcommand{\second}[1]{\textcolor{blue}{#1}}
\newcommand{\best}[1]{\textbf{\textcolor{red}{#1}}}

\usepackage{titletoc}

\theoremstyle{plain}
\newtheorem{theorem}{Theorem}[section]
\newtheorem{proposition}[theorem]{Proposition}
\newtheorem{lemma}[theorem]{Lemma}

\theoremstyle{definition}

\theoremstyle{remark}

\title{Nonparametric Distribution Matching \\for Self-Supervised Whole-Slide Image Condensation}

\author{
Duong M. Nguyen$^{1,\dagger}$,
Trong Nghia Hoang$^{2}$,
Hang Thi Nguyen$^{5}$,
Thanh Trung Huynh$^{3}$, \\
\textbf{Phi Le Nguyen}$^{4}$,
\textbf{Minh N. Do}$^{1,3}$ \\[4pt]
$^{1}$University of Illinois Urbana-Champaign \quad
$^{2}$Washington State University \quad
$^{3}$VinUniversity \\
$^{4}$Hanoi University of Science and Technology \quad
$^{5}$ Hanoi Medical University \\
[4pt]
\texttt{$^\dagger$ Correspondence to \{nmduongg\}@illinois.edu}}

\begin{document}

\maketitle

\begin{abstract}

Histological whole-slide images (WSIs) are central to computational pathology but pose severe computational challenges due to their extremely high resolution, often spanning several gigabytes per slide. To enable scalable learning, existing methods apply self-supervised data condensation to reduce computational cost, but typically rely on heuristic prototype learning and do not explicitly preserve learning-relevant feature distributions for downstream tasks. In response,
we introduce a principled reformulation of WSI condensation as a distribution-matching problem under a fixed representational lens, and develop NICER, a tractable approximation framework based on a nonparametric prior with slide-adaptive capacity. Experiments on five histopathology datasets, together with clinical evaluation from a board-certified pathologist, show that NICER consistently outperforms prior methods, achieving an average accuracy improvement of 7.44\% while offering improved efficiency–accuracy trade-offs, highlighting the benefits of principled, distribution-aware condensation for scalable histological representation learning. Source codes are available in \href{https://github.com/nmduonggg/NICER}{\textit{https://github.com/nmduonggg/NICER}}.

\end{abstract}

\section{Introduction} \label{sec:intro}
Histological whole-slide images (WSIs) are high-resolution digital scans of tissue slides which are central to Computational Pathology (CPath)~\citep{song2024cvpr,Song2023-cq,nguyen2025bridging,xiang2023exploring}. In particular, WSIs have been used as training data for machine learning (ML) systems that enable classification~\citep{xiang2023exploring,shao2021transmil,dong2025wsi}, segmentation~\citep{graham2023segment,guo2023sac}, and survival prediction~\citep{fan2023surv}. Given their enormous resolution (e.g., $100,000 \times 100,000$ pixels) which often amounts to several gigabytes per slide, a WSI cannot be passed as a single training input to a ML system. As a result, most CPath pipelines adopt multiple instance learning (MIL) paradigm, which partitions each slide into thousands of patches (e.g., $>10,000$), embedding them with a pre-trained encoder, and aggregating them into a slide-level representation~\citep{Tang_2023_ICCV,nguyen2025bridging,xiang2023exploring,dong2025wsi}. However, this still results in prohibitive costs of model training and data storage~\citep{jin2025fedwsidd,sacco20onedge,thorstenson2014digital_pathology,lee2026adaptive_compression_wsi,dong2025wsi,shin2025wise,cong2024dataset,alsaafin2024sequential} since each WSI instance is now represented with a massive number of feature patches.

\begin{figure*}[t]
    \centering
    \begin{minipage}[t]{0.58\textwidth} \vspace{0pt}
        \centering
        \includegraphics[width=\linewidth]{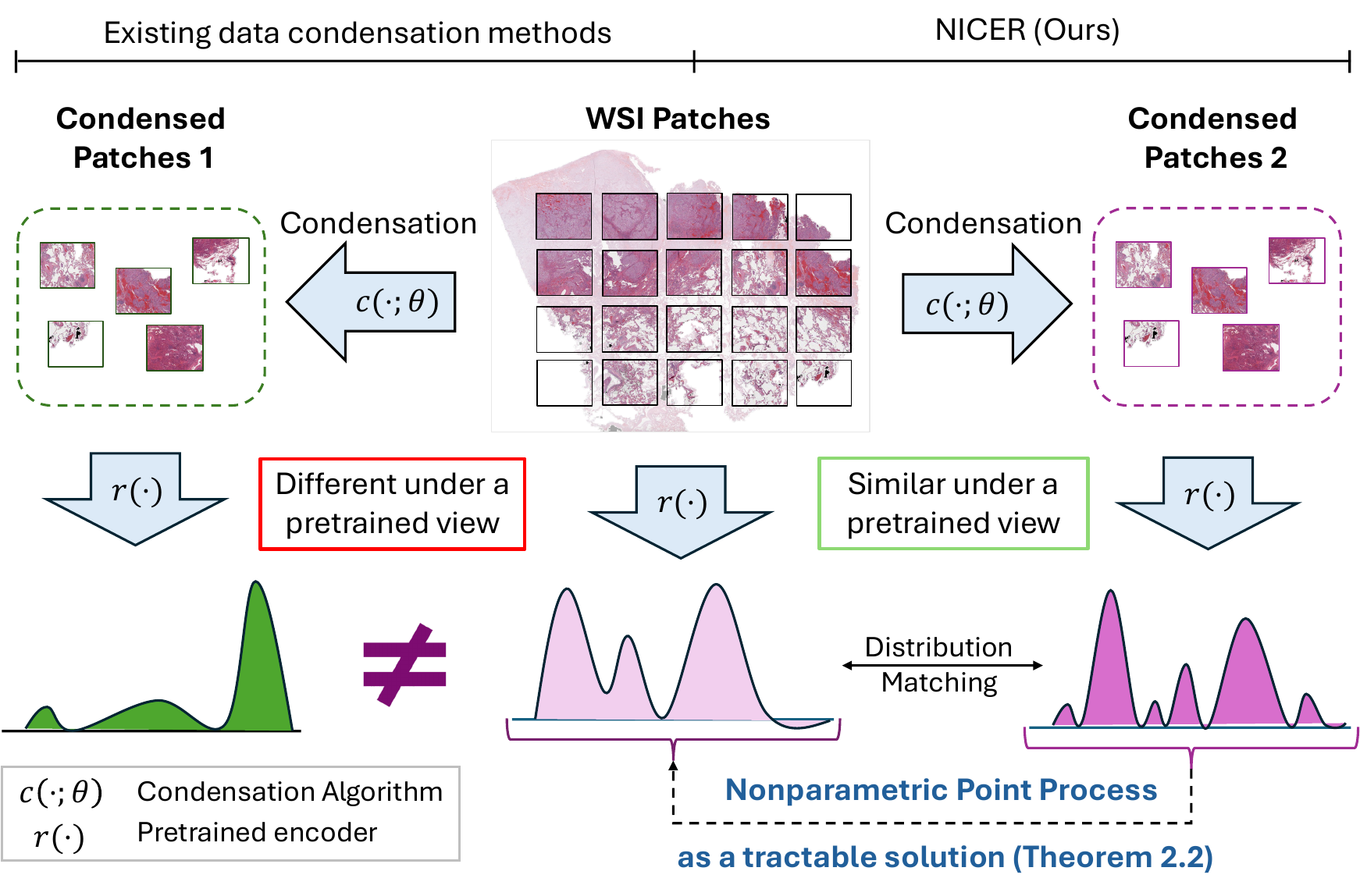}
        \caption{(Left) Existing WSI condensation methods overlook the pretrained representation view. (Right) NICER reformulates WSI condensation as distribution matching under this fixed representation view.}
        \label{fig:motivation2}
    \end{minipage}%
    \hfill%
    \begin{minipage}[t]{0.40\textwidth} \vspace{0pt}
    \centering
    \includegraphics[width=\linewidth]{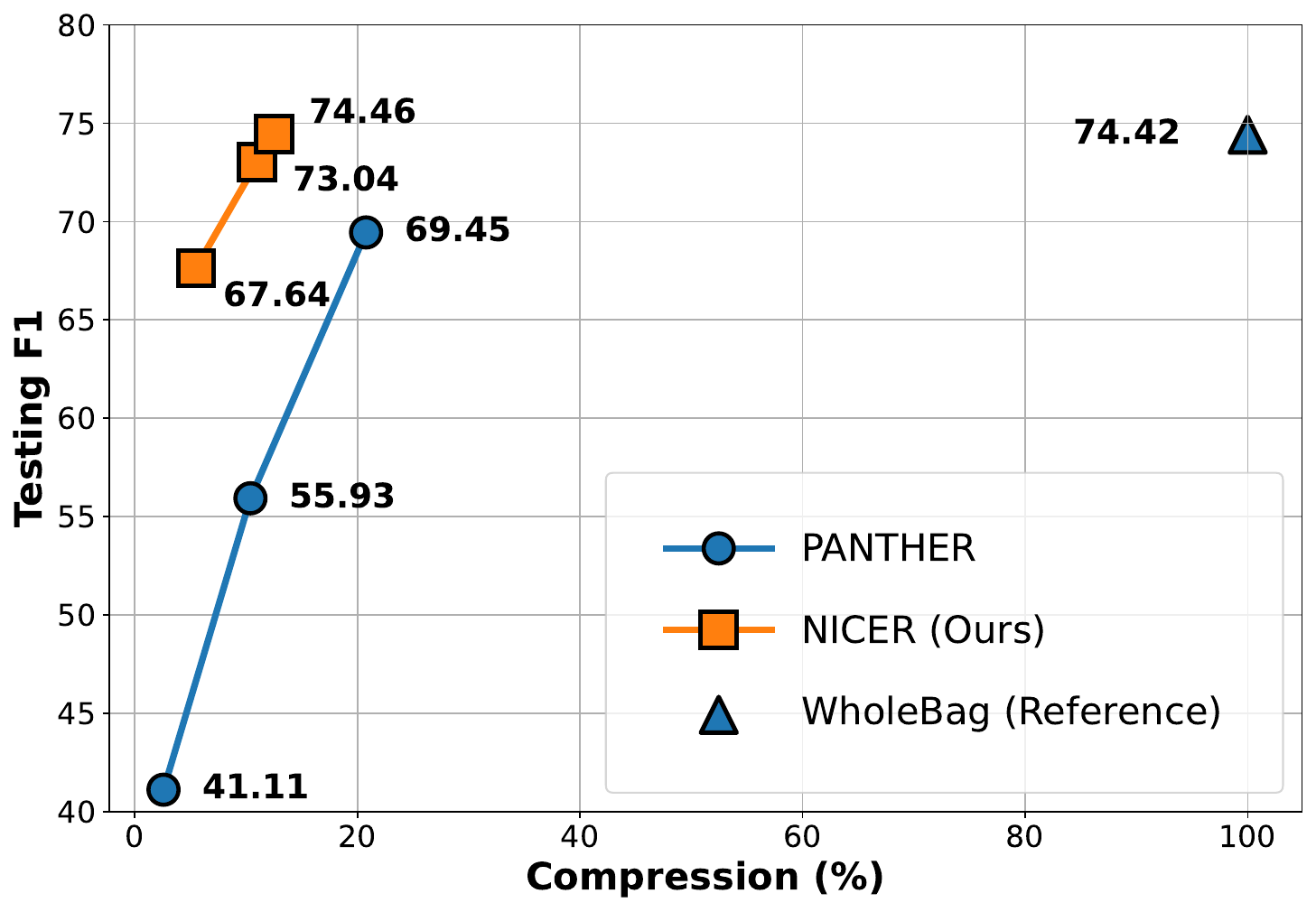}
        \caption{On the same task and model, NICER (ours) achieves better computation-accuracy tradeoff, with higher downstream performance (F1) and smaller compression rate, than PANTHER~\citep{song2024cvpr}, an existing SOTA baseline, by roughly 19\% at a similar compression level. See Sec.~\ref{sec:experiments} for more details.}
        \label{fig:motivation1}
    \end{minipage}%
\vspace{-0.5cm}
\end{figure*}

To keep the training and storage costs manageable, the large set of WSI patches might be further reduced through data condensation~\cite{wang2020dd,song2024cvpr,jin2025fedwsidd,cong2024dataset}. In practice, regardless of the downstream task, most WSI analysis pipelines follow a common paradigm: raw patches are first mapped into a latent feature space by a pretrained encoder, and task-specific heads are then trained on these representations. Therefore, any condensed WSI representation is ultimately interpreted through the latent lens of this fixed encoder, rather than directly in pixel space.

This makes the \textit{encoder-induced feature distribution} critical for WSI condensation. Specifically, even if different condensation strategies preserve comparable information in principle, they may induce substantially different distributions after feature extraction. Condensations aligned with this fixed representation viewpoint are more likely to produce features suitable for downstream learning, whereas others may be information-preserving in the abstract but poorly matched to the geometry used by subsequent models (see Fig.~\ref{fig:motivation2}).

\textbf{Existing Approaches.}  
Despite growing interest in histological imaging, existing data condensation methods largely overlook this encoder-induced information preservation. These methods can be broadly categorized into two paradigms. First, task-specific approaches leverage differential feedback from the solution head to guide condensation or distillation, explicitly optimizing the condensed representation for a particular downstream objective~\citep{deng2022remember,zheng2024leveraginghierarchicalfeaturesharing,jin2025fedwsidd,cong2024dataset}. In contrast, task-agnostic approaches~\citep{song2024cvpr, vu2023h2t,alsaafin2024sequential,claudio2024} aim to condense WSIs independently of any task, typically relying on clustering or unsupervised prototype learning to summarize patch distributions. While the latter paradigm is often more suitable for WSIs due to its reusability, it still overlooks the fact that the distilled data will still be processed by a pre-trained feature extractor with a fixed inductive bias.



\textbf{Fundamental Gap.}
To this end, existing WSI condensation methods lack a principled formulation that explicitly accounts for condensed information preservation under a fixed representation lens. Prior approaches typically rely on heuristic condensation strategies or ad-hoc prototype selection, without formally characterizing how condensed representations should preserve the induced feature distributions. Addressing this gap requires reframing condensation as a \textit{distribution-matching} problem: rather than learning a fixed set of  prototypes from raw patch embeddings, condensation should aim to approximate the observed representation distributions in a controlled manner.

\textbf{Our Contributions.} Motivated by this reformulation, we introduce a theoretically grounded condensation framework that formalizes WSI condensation as a distribution-matching problem and provides practical mechanisms to realize this objective at scale. Our contributions are as follows:

\textbf{(1)}~We reformulate WSI condensation as a distribution-matching problem under a fixed representation lens, moving beyond heuristic prototype learning toward preserving the latent feature distribution induced by a frozen encoder. Theoretically, we connect this objective to a MAP solution under a point-process prior, yielding a principled nonparametric formulation that characterizes the equivalent solution class (see Sec.~\ref{sec:prelim}).

\textbf{(2)}~Our theoretical reformulation motivates NICER, a nonparametric WSI condensation framework that enables tractable optimization of the intractable distribution-matching objective. NICER adopts a learning-based approximation that decouples local pattern coverage from global concept abstraction, providing a principled and efficient mechanism for adaptive-capacity condensation (see Sec.~\ref{sec:method}). 

\textbf{(3)}~Empirically, we evaluate NICER on cancer subtyping and survival prediction across five benchmark datasets, further supported by pathological assessment from a board-certified pathologist. The results show that NICER consistently outperforms existing baselines, achieving better performance-computation tradeoffs while preserving clinically meaningful histological structure (see Sec.~\ref{sec:experiments}).

\section{Theoretical Framework} \label{sec:prelim}
\textbf{Background.}
In computational pathology, a WSI $w$ is represented by a set of $N$ patches
$\{x_i\}_{i=1}^N \subset \mathcal{X}$, where $\mathcal{X}$ is patch space, and is modeled as a random sample $w\sim p(w)$.
Let $r:\mathcal{X}\to\mathbb{R}^d$ denote a fixed pretrained {patch encoder}.
We view $r(\cdot)$ as defining a \emph{fixed representational viewpoint}.

Under this viewpoint, each WSI induces an empirical distribution over patch embeddings, while a condensation mechanism $c(\cdot;\theta)$, parameterized by $\theta$, maps each WSI to a reduced patch set
$c(w;\theta)=\{\tilde x_k\}_{k=1}^K$ with $K\ll N$, inducing the empirical and condensed embedding measures:
\begin{equation}
\label{eq:mu_w_def}
\mu_w \;\triangleq\; \frac1N \sum_{i=1}^N \delta_{r(x_i)}
\;\in\; \mathcal{P}_2(\mathbb{R}^d), \quad \nu_{w,\theta} \;\triangleq\; \frac1K \sum_{k=1}^K \delta_{r(\tilde x_k)}
\;\in\; \mathcal{P}_2(\mathbb{R}^d).
\end{equation}
where $\delta$ is Dirac delta function and $\mathcal{P}_2(\mathbb{R}^d)$ denotes the space of probability measures with finite second moment.
Mixing over slides yields the {global} dataset-level patch feature distributions:
\begin{equation}
\label{eq:global_pr_ptheta}
p_r \;\triangleq\; \mathbb{E}_{w\sim p(w)}[\mu_w],
\qquad
p_\theta \;\triangleq\; \mathbb{E}_{w\sim p(w)}[\nu_{w,\theta}],
\end{equation}

\textbf{Problem Formulation.}~We aim to learn a condensation mechanism $c(\cdot; \theta)$ whose induced representation distribution matches real WSIs under the same pretrained viewpoint:
\begin{equation}
\label{eq:dist_matching}
\min_{\theta}\; D\!\left(p_r,\;p_\theta\right),
\end{equation}
where $D(\cdot,\cdot)$ is a discrepancy between probability measures on $\mathbb{R}^d$.
In this work, we consider squared 2-Wasserstein distance~\cite{courty2016ot,peyre2019computational}, i.e., $D \equiv \mathrm{OT}(\cdot,\cdot)$ and leave other distances to the future work.

\noindent\textbf{Remark.}
Eq.~\eqref{eq:dist_matching} formalizes WSI condensation as \emph{distribution matching under a fixed pretrained lens}.
Different condensed patch sets may preserve comparable semantic content, yet induce different embedding distributions through $r(\cdot)$;
only condensations aligned with this representational viewpoint achieve small distribution discrepancy.

The distribution-matching objective alone does not guarantee condensation in an operational sense. 
We make this objective operational by constraining \(c(\cdot;\theta)\) to select patches whose embeddings lie in a shared global {prototype} set 
\(\Phi=\{\phi_t\}_{t=1}^T \subset \mathbb{R}^d\), which serves as the condensed dataset in the representation space. 
This induces a condensed distribution \(p_{\Phi,\theta}\) and yields a penalized objective,
\begin{equation}
\label{eq:penalized_wass}
\min_{\Phi,\theta}\;
\mathrm{OT}\!\left(p_r,\;p_{\Phi,\theta}\right)
+\lambda|\Phi|,
\qquad \lambda>0.
\end{equation}
\textbf{Solution Derivation.}~Directly solving Eq.~\eqref{eq:penalized_wass} is challenging due to the global OT coupling between mixture distributions across slides, together with variable-cardinality prototype inference over a large candidate domain.
To obtain a scalable solution, we optimize a more tractable surrogate that upper-bounds the global Wasserstein discrepancy by an expectation of per-slide discrepancies. This surrogate decomposes naturally across slides, enabling independent per-slide inference while remaining aligned with the original objective, as shown in Proposition~\ref{prop:global_to_slide} (proof in App.~\ref{app:prop_proof}).

\begin{proposition}[Global-to-Per-Slide Reduction]
\label{prop:global_to_slide}
Let $\mu_w,\nu_{w,\theta}\in\mathcal{P}_2(\mathbb{R}^d)$ denote the real and condensed
within-slide embedding measures defined in
Eq.~\eqref{eq:mu_w_def} and let
$p_r$ and $p_\theta$ be the corresponding global representation distributions
defined in Eq.~\eqref{eq:global_pr_ptheta}.
Then the OT discrepancy satisfies
\begin{equation}
\label{eq:global_le_slide}
\mathrm{OT}\!\left(p_r,\;p_\theta\right)
\;\le\;
\mathbb{E}_{w\sim p(w)}\!\left[
\mathrm{OT}\!\left(\mu_w,\nu_{w,\theta}\right)
\right].
\end{equation}
Moreover, if condensation is restricted to slide-specific finite configurations from the representation space
$\Phi_w\subset\mathbb{R}^d$, yielding condensed measures $\nu_{\Phi_w,\theta}$, then
\begin{equation}
\label{eq:global_penalized_le_slide}
\begin{aligned}
&\mathrm{OT}\!\left(p_r,\;p_{\{\Phi_w\},\theta}\right)
+\lambda\,\mathbb{E}_{w}[|\Phi_w|]
\; \le\;
\mathbb{E}_{w}\!\left[
\mathrm{OT}\!\left(\mu_w,\nu_{\Phi_w,\theta}\right)
+\lambda|\Phi_w|
\right].
\end{aligned}
\end{equation}
\end{proposition}
\textbf{Remark.}~Proposition~\ref{prop:global_to_slide} shows that, by the joint convexity of OT, transporting averaged measures is no harder than averaging per-slide transport costs. This motivates the per-slide objectives in Eq.~\eqref{eq:global_penalized_le_slide}, allowing the dataset-level problem over $\Phi$ to be decomposed into independent slide-level problems over $\{\Phi_w\}_{w\sim p(w)}$. While both formulations preserve slide-specific traits, the dataset-level prototypes may merge shared patterns across slides too aggressively. Thus, the slide-specific surrogate is preferable in MIL, where each WSI is processed independently at inference.

Despite this reduction, Eq.~\eqref{eq:global_penalized_le_slide} remains an optimization problem over variable-cardinality configurations and continuous parameters. We sidestep this intractability via Theorem~\ref{thm:ppp_map}, which characterizes solutions for Eq.~\eqref{eq:global_penalized_le_slide} as MAP estimates under a Gibbs posterior induced by a Poisson point process (PPP) prior (proof in App.~\ref{app:thm_proof}). The PPP naturally models finite, unordered sets with unknown cardinality and promotes sparsity through a linear size penalty, making it well suited for prototype configurations. Importantly, the PPP prior is placed on prototypes, which capture distinct prototypical information, rather than on the strongly correlated patches. Hence, we note that our PPP prior consistently aligns with standard assumption about cross-patch correlations in MIL settings~\cite{nguyen2025bridging,abmil}.

\begin{theorem}[PPP-MAP interpretation with OT]
\label{thm:ppp_map}
Let $\Omega\subset\mathbb{R}^d$ be a bounded representation domain and let $\Pi_\rho(\Phi)$ be a homogeneous Poisson point process (PPP) prior over finite prototype configurations $\Phi_w\subset\Omega$ with intensity $\rho \in (0, 1)$.
Let $\pi_0(\theta)$ be any prior over condensation parameters $\theta$.
For $\varepsilon>0$ and $\beta>0$, consider the Gibbs posterior
\begin{equation}
\label{eq:gibbs_posterior_entropic}
\begin{aligned}
&\pi(\Phi_w,\theta \mid \mu_w) \propto
\exp\Big(-\beta\,\mathrm{OT}\left(\mu_w,\nu_{\Phi_w,\theta}\right)\Big)
\Pi_\rho(\Phi_w)
\pi_0(\theta),
\end{aligned}
\end{equation}
where $\mathrm{OT}(\cdot,\cdot)$ denotes the optimal transport
between probability measures with squared-Euclidean cost in representation space.
Then any MAP estimator
\(
(\Phi_w^\star,\theta^\star)\in\arg\max_{\Phi_w,\theta}\pi(\Phi_w,\theta\mid \mu_w)
\)
is equivalently a minimizer of
\begin{equation}
\label{eq:map_objective_entropic}
\min_{\Phi_w,\theta}\;
\mathrm{OT}\!\left(\mu_w,\;\nu_{\Phi_w,\theta}\right)
+\lambda|\Phi_w|
-\tfrac{1}{\beta}\log \pi_0(\theta),
\end{equation}
where $\lambda>0$ is determined by the PPP intensity $\rho$ and $\beta$.
In particular, when $\pi_0(\theta)$ is uniform,
MAP inference under a PPP prior corresponds to minimizing a penalized OT objective.
\end{theorem}
\textbf{Intuition.}~Proposition~\ref{prop:global_to_slide} shows that the original objective in Eq.~\ref{eq:penalized_wass} is upper-bounded by the slide-level objective in Eq.~\ref{eq:global_penalized_le_slide}.
Theorem~\ref{thm:ppp_map} further gives this slide-level distribution-matching objective a nonparametric probabilistic interpretation through a PPP prior over $\Phi$. \textit{In particular, it shows that the MAP solutions induced by the PPP prior coincide with minimizers of the proposed cardinality-penalized condensation objective.} This provides a probabilistic perspective on prototypical condensation: rather than treating Eq.~\ref{eq:penalized_wass} solely as a direct OT problem, one can equivalently formulate a MAP inference problem under a nonparametric prior, which motivates a tractable approximation as below.

\textbf{Learning-based Approximation.}
While direct PPP-MAP inference remains computationally intractable on the large empirical support $\mathcal{H}_w=\{r(x_i)\}_{i=1}^N$ of $\mu_w$ given a slide $w$, the flexibility of this probabilistic formulation allows us to introduce latent variables while preserving the distribution-matching structure. We therefore use a proxy pattern set $\mathcal{Z}_w$, with $|\mathcal{Z}_w| \ll |\mathcal H_w|$, as a reduced, data-adaptive support for $\mu_w$, and solve an approximate joint MAP problem over $\mathcal{Z}_w$ and $\Phi_w$:
\begin{equation}
\label{eq:log_joint_prob}
\max_{\mathcal{Z}_w, \Phi_w} \; \log \mathbb{P}(\mathcal{H}_w, \mathcal{Z}_w \mid \Phi_w) = \max_{\mathcal{Z}_w, \Phi_w} \;
\log \mathbb{P}(\mathcal{H}_w \mid \mathcal{Z}_w)
+
\log \mathbb{P}(\mathcal{Z}_w \mid \Phi_w),
\end{equation}
where $\mathbb{P}(\mathcal{H}_w \mid \mathcal{Z}_w)$ encourages $\mathcal{Z}_w$ to summarize the empirical embedding distribution, while $\mathbb{P}(\mathcal{Z}_w \mid \Phi_w)$ approximates the PPP-induced nonparametric mechanism over $\Phi_w$ to support adaptive sparsity (see App.~\ref{app:nicer_nonparametric_prior}).
This formulation relies on the conditional independence assumption $
\mathbb{P}(\mathcal{H}_w \mid \mathcal{Z}_w, \Phi_w)=
\mathbb{P}(\mathcal{H}_w \mid \mathcal{Z}_w),$
yielding a tractable approximation to the per-slide PPP-MAP inference. 




Optimization alternates between inferring $\mathcal{Z}_w$ as a summary of $\mathcal{H}_w$ and solving (parameterized) PPP-MAP formulation over $\Phi_w$ on this support. This two-stage structure allows learned prototypes to revisit the full empirical distribution through $\mathcal{Z}_w$, improving both efficiency and information preservation. We note that the condensation parameters $\theta$ in Eq.~\ref{eq:global_pr_ptheta} are instantiated as slide-specific variables $\mathcal{Z}_w$ and learnable parameters in our MAP formulation (detailed in Sec.~\ref{sec:condensation}).
For simplicity, we omit the slide index $w$ in the following derivations.


\begin{figure*}
    \centering    \includegraphics[width=\linewidth]{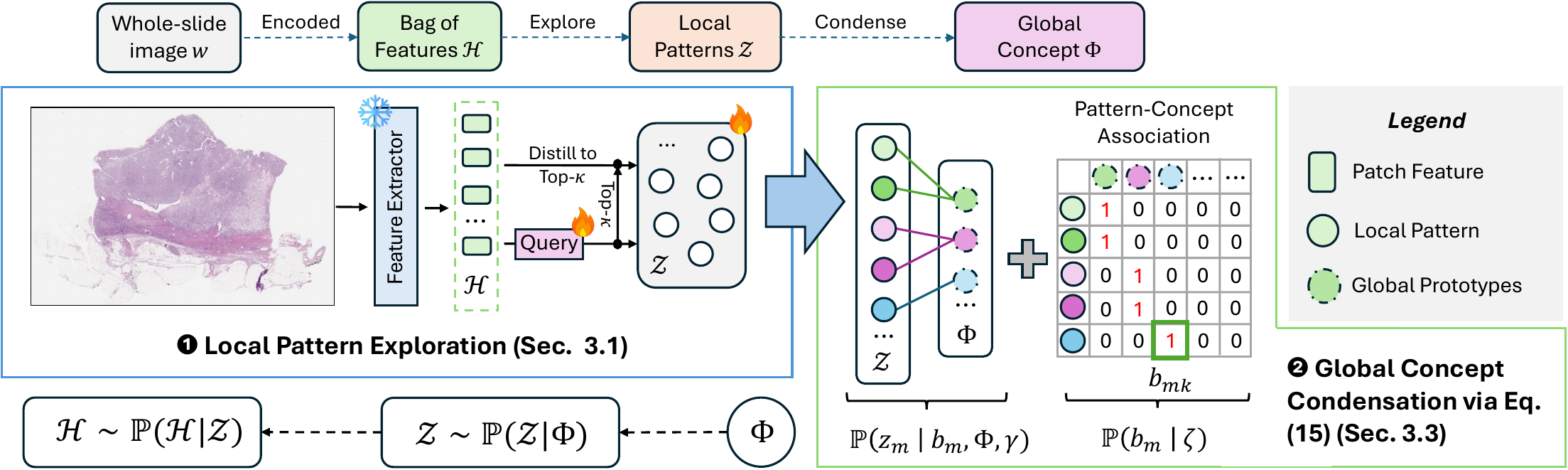}
    \caption{\textbf{Overview of NICER.} NICER provides a learning-based approximation by decoupling patterns $\mathcal{Z}$ and concepts $\Phi$, learned via two stages: \textit{pattern exploration} to capture diverse information, and \textit{condensation} to merge redundancies into compact concepts driven by a nonparametric prior.}
    \vspace{-0.5cm}
    \label{fig:overview}
\end{figure*}

\section{A Tractable Condensation Approximation} \label{sec:method}

This section presents \textbf{NICER}, a \textbf{N}onparametr\textbf{IC} self-sup\textbf{ER}vised WSI condensation framework designed to preserve slide-level representation distributions under a fixed encoder. NICER operationalizes our theoretical formulation through a two-stage procedure: \textit{pattern exploration}, which summarizes empirical patch embeddings $\mathcal{H}$ into a compact set of local patterns $\mathcal{Z}$ by optimizing $\log \mathbb{P}(\mathcal{H}\mid\mathcal{Z})$ (Sec.~\ref{sec:distillation}); and \textit{concept condensation}, which performs nonparametric abstraction of these patterns into a latent concept set $\Phi$ (Fig.~\ref{fig:overview}). 
We operationalize the PPP-induced sparsity interpretation by introducing assignment variables, which allow redundant concepts to collapse and the effective model complexity to adapt to slide-specific variability (Sec.~\ref{sec:condensation}).


\subsection{Local Pattern Exploration} \label{sec:distillation}

\textbf{Intuition.}
We model $\mathcal{Z}$ as a compact summary of the feature bag $\mathcal{H}$ that preserves important structure while reducing redundancy. Each feature $h_i$ is explained by a nearby pattern in $\mathcal{Z}$ via a Gaussian model,
$h_i \sim \mathcal{N}\!\left(\mu_i(\mathcal{Z}), \sigma^2 \mathbf{I}\right)$,
where $\mu_i(\mathcal{Z})$ selects a representative pattern for $h_i$. Learning $\mathcal{Z}$ thus amounts to maximizing $\log \mathbb{P}(\mathcal{H}\mid\mathcal{Z})$, which corresponds to a similarity-based retrieval-style approximation of the feature distribution (see App.~\ref{app:derive_HZ} for more details).

\textbf{Design and Learning.}
We implement this objective as a retrieval-based pattern selection process using cosine similarity $
\Gamma(h_i, z_m) \triangleq \langle h_i, z_m \rangle $.
To allow flexibility while maintaining compactness, each feature distributes its mass over its top-$\kappa$ nearest patterns ($\kappa \ll M$). This yields a stable and efficient approximation of $\mathcal{H}$ by $\mathcal{Z}$, which we treat as an intermediate latent representation.

\subsection{Global Concept Condensation}  \label{sec:condensation}

\textbf{Intuition. }~This section derives the condensation term in Eq.~\ref{eq:log_joint_prob} and formalizes the construction of the concept set $\Phi$ from the intermediate patterns $\mathcal{Z}$. We model this stage as a nonparametric point process, where $\mathcal{Z}$ is generated from a latent concept set $\Phi$ with unknown capacity prior to observation of $\mathcal{Z}$. To jointly learn $\Phi$ and its assignments to patterns in $\mathcal{Z}$, we introduce a \textit{pattern prior} that governs pattern-concept relationships and an \textit{assignment prior} that captures their association structure.

\textbf{Pattern Prior.}
Under the squared-Euclidean ground cost used in our OT formulation, matching a local pattern $z_m$ to a concept $\phi_k$ incurs the quadratic cost $\|z_m-\phi_k\|_2^2$ in representation space~\cite{peyre2019computational}. Motivated by this geometry, we model the pattern-concept relation with a Gaussian likelihood centered at a single concept prototype, whose log-likelihood contains the same squared-distance term. Specifically, each pattern $z_m \in \mathcal{Z}$ is associated with exactly one concept $\phi_k \in \Phi$ and we model
\begin{equation}
\label{eq:N_old}
z_m \sim \mathcal{N}\!\big(\psi_m(\Phi),\, \Sigma_m(\Phi,\gamma)\big),
\end{equation}
where $\psi_m(\Phi)$ selects the responsible concept and $\Sigma_m(\Phi,\gamma)$ specifies a (learnable) covariance.

\textbf{Assignment Variable.}~To enforce a single-concept assignment, we introduce one-hot variables $b_m=(b_{m1},\ldots,b_{mK})$ with $b_{mk}\in\{0,1\}$ and $\sum_{k=1}^K b_{mk}=1$, and define $
\psi_m \triangleq \sum_{k=1}^K b_{mk}\,\phi_k.$
For computational efficiency, we restrict $\Sigma(\Phi; \gamma)$ to be diagonal, which corresponds to learning a per-dimension adaptive metric in representation space and is consistent with the squared-Euclidean cost underlying OT, yielding
\begin{equation}
\label{eq:N_new}
z_m \mid b_m \sim \mathcal{N}\!\Big(\psi_m,\; \mathrm{diag}\big(\vartheta(\psi_m;\gamma)\big)\Big),
\end{equation}
where $\vartheta(\cdot;\gamma)$ is a neural function predicting per-dimension variances. Here, $b_{mk}=1$ indicates that pattern $z_m$ is explained by concept $\phi_k$.

\textbf{Assignment Prior.}
Given the pattern prior in Eqs.~\eqref{eq:N_old},\eqref{eq:N_new} and the assignment $b=\{b_m\}_{m=1}^M$, the condensation term in Eq.~\eqref{eq:log_joint_prob} can be expressed over the joint distribution $\mathbb{P}(\mathcal{Z}, b \mid \Phi)$ as
\begin{equation} \label{eq:assignment_prior}
\begin{aligned}
    &\log \mathbb{P}(\mathcal{Z}, b \mid \Phi) = \log \left \{\prod_{m=1}^M \mathbb{P}(z_m, b_m \mid \Phi, \gamma, \zeta)\right\} =\sum_{m=1}^{M} \log \mathbb{P}(z_m \mid b_m, \Phi, \gamma) + \log \mathbb{P}(b_m \mid \zeta),
\end{aligned}
\end{equation}
where $\mathbb{P}(b_m \mid \zeta)$ defines a prior over pattern-concept assignments with parameters $\zeta$.
To enforce each pattern $z_m \in \mathcal{Z}$ associates with exactly one concept, we impose a categorical prior over $b_m$,
\begin{equation}
\label{eq:cate_b}
\mathbb{P}(b_m \mid \zeta)
\triangleq
\prod_{k=1}^K
\left(
\frac{\exp(\alpha(\phi_k;\zeta))}{\sum_{k'} \exp(\alpha(\phi_{k'};\zeta))}
\right)^{b_{mk}},
\end{equation}
where $\alpha(\cdot)$ is a neural function parameterized by $\zeta$.
This assignment prior enforces a single-concept explanation for each pattern while allowing many-to-one aggregation across patterns. The learnable logits $\alpha(\phi_k;\zeta)$ enable adaptive pattern-concept allocation that reflects the diversity of $\mathcal{Z}$.

\subsection{Condensation Learning} \label{3.2 sec: proto condensation}

Given the nonparametric pattern-concept model above, our condensation term in Eq.~\ref{eq:log_joint_prob} now reduces to an optimization problem with all newly defined parameters,
\begin{equation}
\label{eq:psi_final}
\begin{aligned}
     &\max_{\Phi, \gamma, \zeta, b}  \sum_{m=1}^M \log \mathbb{P}( z_m, b_m \mid \gamma, \zeta, \Phi )= \max_{\Phi, \gamma, \zeta, b} \sum_{m=1}^M \bigg\{\log \mathbb{P}( z_m \mid b_m , \gamma, \Phi ) + \log \mathbb{P}( b_m \mid \zeta ) \bigg\},
\end{aligned}
\end{equation} 
which is directly computable when using Eq.~\ref{eq:N_new} and Eq.~\ref{eq:cate_b} (see App.~\ref{app:lemma_proof} for more details). Solving Eq.~\ref{eq:psi_final} is however not trivial due to its mixed set of discrete/continuous variables. To sidestep this intractability, we instead solve
Eq.~\ref{eq:psi_final} via alternating between (1) optimizing $(\gamma, \zeta, \Phi)$ while fixing $b$; and (2) optimizing $b$ given $(\gamma, \zeta, \Phi)$.

\textbf{Updating $(\Phi, \gamma, \zeta)$ given $b$.}~With assignments fixed, Eq.~\eqref{eq:psi_final} reduces to a continuous optimization problem derived directly from the Gaussian form in Eq.~\eqref{eq:N_new} and the categorical prior in Eq.~\eqref{eq:cate_b}. This subproblem is differentiable and can be optimized using gradient-based methods.

\textbf{Update $b$ given $(\gamma, \zeta, \Phi)$.}~Optimizing the discrete assignment variables is more challenging due to the one-hot constraint on each $b_m$. However, this structure enables a key simplification: since each $b_m$ selects exactly one concept, the log-joint objective becomes linear in $b_m$ when $(\Phi, \gamma, \zeta)$ are fixed, as shown by Eqs.~\eqref{eq:N_new},\eqref{eq:cate_b} and Lemmas~\ref{lemma:R1},\ref{lemma:R2}. In particular, Eq.~\ref{eq:psi_final} becomes:
\begin{equation} \label{eq:final_linear_form}
\begin{aligned}
    &b^\ast = \argmax_b \bigg\{ \sum_{m=1}^M \log \mathbb{P} \big( z_m, b_m \mid \gamma, \zeta, \Phi \big) \bigg\}  = \argmax_b \bigg\{ R_1 (b) + R_2(b) \bigg\}, \\
    \text{where} \quad & R_1(b) = \sum_{m=1}^M \sum_{k=1}^K b_{mk} \cdot \log \mathcal{N} \left( z_m \mid \phi_k; \text{diag}\left( \vartheta \left( \phi_k; \gamma \right) \right) \right) \\
    & R_2(b) = \sum_{m=1}^M \sum_{k=1}^K b_{mk} \cdot \log \bigg( \frac{\exp\big( \alpha(\phi_k; \zeta) \big)}{\sum_{k'} \exp \big( \alpha(\phi_{k'}; \zeta) \big)} \bigg), 
\end{aligned}
\end{equation}
which is a weighted linear optimization task. Here, $R_1(b)$ is derived  from our pattern prior while $R_2(b)$ originates from the categorical assignment prior. Since the resulting objective decomposes across patterns, we iteratively optimize each $b_m$ while holding the remaining assignments fixed. This reduces Eq.~\ref{eq:final_linear_form} to a maxima search that can be solved with linear complexity $\mathcal{O}(M)$ per iteration. 
When solving for $b$ in Eq.~\ref{eq:final_linear_form}, concepts with no assignment are removed as redundant. We note that this active-set update operationalizes the PPP-induced preference for smaller configurations, rather than a separate heuristic pruning step (see App.~\ref{app:nicer_nonparametric_prior}). Pseudocode for NICER can be found in App.~\ref{app:pseudocode}.

\section{Experiments} \label{sec:experiments}
This section empirically evaluates NICER on five datasets across two key pathology tasks: cancer subtyping and survival prediction.  Dataset details and baselines are provided in Sec.~\ref{sec:exp_settings}, while experimental results are in Sec.~\ref{sec:main_results}.

\begin{table*}[]
\centering
\caption{Performance of baselines on Condensation tasks. \best{The best} and \second{second-best} results are highlighted in \best{bold red}, and \second{blue}. Here, \textit{Whole Bag} serves as a full-feature reference, not a baseline.} \label{tab:main-results-dd}
\resizebox{\textwidth}{!}{%
\begin{tabular}{@{}l|c|cccccc|ccc@{}}
\toprule
\multicolumn{1}{c|}{\multirow{3}{*}{Method}} &
  \multirow{3}{*}{Decoder} &
  \multicolumn{6}{c|}{Cancer Subtyping} &
  \multicolumn{3}{c}{Survival Prediction} \\ \cmidrule(l){3-11} 
\multicolumn{1}{c|}{} &
   &
  \multicolumn{3}{c|}{PANDA} &
  \multicolumn{3}{c|}{NSCLC} &
  LUAD &
  BRCA &
  CPTAC \\ \cmidrule(l){3-11} 
\multicolumn{1}{c|}{} &
   &
  Kappa &
  Accuracy &
  \multicolumn{1}{c|}{F1} &
  Kappa &
  Bal. Acc. &
  F1 &
  C-Index &
  C-Index &
  C-Index \\ \midrule
Whole Bag &
  \multirow{8}{*}{ABMIL} &
  91.93 ± 0.48 &
  76.21 ± 1.53 &
  \multicolumn{1}{c|}{76.37 ± 1.38} &
  90.31 ± 1.65 &
  94.52 ± 1.19 &
  95.19 ± 0.78 &
  62.12 ± 1.27 &
  78.52 ± 3.82 &
  48.26 ± 4.87 \\
DeepSets &
   &
  57.26 ± 38.13 &
  51.60 ± 17.98 &
  \multicolumn{1}{c|}{46.42 ± 23.91} &
  79.51 ± 1.80 &
  89.82 ± 0.88 &
  89.73 ± 0.91 &
  \second{59.89 ± 5.34} &
  49.23 ± 3.59 &
  49.80 ± 7.96 \\
ProtoCount &
   &
  0.83 ± 9.55 &
  24.24 ± 1.21 &
  \multicolumn{1}{c|}{11.77 ± 1.28} &
  10.69 ± 3.54 &
  55.34 ± 1.78 &
  53.86 ± 3.15 &
  51.91 ± 5.75 &
  56.47 ± 12.03 &
  51.72 ± 8.97 \\
H2T &
   &
  \second{75.03 ± 1.08} &
  \second{53.91 ± 1.10 }&
  \multicolumn{1}{c|}{\second{50.66 ± 1.01}} &
  79.45 ± 1.80 &
  89.67 ± 0.88 &
  89.72 ± 0.90 &
  51.83 ± 2.59 &
  45.86 ± 4.29 &
  45.34 ± 0.67 \\
OT &
   &
  41.92 ± 11.16 &
  35.50 ± 3.22 &
  \multicolumn{1}{c|}{29.98 ± 4.49} &
  80.79 ± 4.69 &
  90.39 ± 2.36 &
  90.37 ± 2.37 &
  54.07 ± 3.71 &
  66.61 ± 6.30 &
  \second{57.28 ± 2.27} \\
InfiniteGPFA &
   &
  0.00 ± 0.00 &
  26.85 ± 0.00 &
  \multicolumn{1}{c|}{11.36 ± 0.00} &
  6.61 ± 2.49 &
  53.29 ± 1.25 &
  50.61 ± 3.93 &
  50.48 ± 4.35 &
  46.65 ± 10.52 &
  52.18 ± 7.90 \\
PANTHER &
   &
  65.52 ± 12.62 &
  43.77 ± 0.76 &
  \multicolumn{1}{c|}{42.52 ± 1.32} &
  \second{83.98 ± 2.39} &
  \second{92.02 ± 1.18} &
  \second{92.01 ± 1.21} &
  58.92 ± 2.21 &
  \second{72.07 ± 6.17} &
  51.26 ± 2.61 \\ \cmidrule(r){1-1} \cmidrule(l){3-11} 
\textbf{NICER (Ours)} &
   &
  \best{90.98 ± 0.97} &
  \best{73.52 ± 1.27} &
  \multicolumn{1}{c|}{\best{73.22 ± 1.55}} &
  \best{89.73 ± 1.81} &
  \best{94.84 ± 0.92} &
  \best{94.87 ± 0.91} &
  \best{64.67 ± 2.11} &
  \best{75.69 ± 2.02} &
  \best{62.36 ± 2.06} \\ \midrule \midrule
Whole Bag &
  \multirow{8}{*}{DSMIL} &
  90.93 ± 0.43 &
  74.05 ± 1.38 &
  \multicolumn{1}{c|}{74.11 ± 1.31} &
  84.60 ± 4.16 &
  92.27 ± 2.08 &
  92.30 ± 2.08 &
  63.32 ± 4.02 &
  67.96 ± 7.12 &
  47.72 ± 4.99 \\
DeepSets &
   &
  \second{85.53 ± 3.36 }&
  \second{62.26 ± 1.77 }&
  \multicolumn{1}{c|}{\second{62.88 ± 1.73}} &
  25.30 ± 36.76 &
  62.68 ± 18.46 &
  59.10 ± 20.77 &
  56.86 ± 6.41 &
  48.47 ± 15.91 &
  36.86 ± 2.13 \\
ProtoCount &
   &
  14.78 ± 8.52 &
  22.52 ± 2.11 &
  \multicolumn{1}{c|}{14.47 ± 0.53} &
  -0.90 ± 13.16 &
  49.52 ± 6.63 &
  44.24 ± 10.37 &
  52.71 ± 4.47 &
  41.13 ± 8.69 &
  47.71 ± 2.98 \\
H2T &
   &
  60.56 ± 1.57 &
  41.83 ± 0.84 &
  \multicolumn{1}{c|}{34.73 ± 5.23} &
  44.90 ± 3.95 &
  72.61 ± 2.03 &
  70.10 ± 2.64 &
  50.00 ± 0.00 &
  49.82 ± 0.26 &
  45.31 ± 2.90 \\
OT &
   &
  55.06 ± 7.77 &
  37.66 ± 3.74 &
  \multicolumn{1}{c|}{36.60 ± 3.94} &
  48.26 ± 8.11 &
  74.20 ± 4.00 &
  73.52 ± 4.49 &
  \second{58.17 ± 4.11} &
  \second{70.84 ± 0.23} &
  \second{49.36 ± 2.41} \\
InfiniteGPFA &
   &
  17.72 ± 9.52 &
  24.76 ± 2.19 &
  \multicolumn{1}{c|}{20.56 ± 2.03} &
  -2.70 ± 12.64 &
  48.62 ± 6.31 &
  46.52 ± 7.29 &
  46.25 ± 9.89 &
  46.16 ± 6.69 &
  45.27 ± 9.69 \\
PANTHER &
   &
  71.74 ± 2.46 &
  46.09 ± 3.02 &
  \multicolumn{1}{c|}{46.32 ± 3.28} &
  \second{75.05 ± 8.12} &
  \second{87.60 ± 4.04} &
  \second{87.46 ± 4.13 }&
  54.15 ± 4.88 &
  {69.96 ± 2.49 }&
  46.66 ± 6.29 \\ \cmidrule(r){1-1} \cmidrule(l){3-11} 
\textbf{NICER (Ours)} &
   &
  \best{87.81 ± 0.69} &
  \best{66.07 ± 1.73} &
  \multicolumn{1}{c|}{\best{65.19 ± 1.65}} &
  \best{84.61 ± 1.56} &
  \best{92.30 ± 0.74} &
  \best{91.34 ± 0.79} &
  \best{62.12 ± 1.41} &
  \best{74.53 ± 1.73 }&
  \best{56.71 ± 1.18} \\ \midrule \midrule
Whole Bag &
  \multirow{8}{*}{ILRA} &
  93.46 ± 0.24 &
  78.90 ± 1.76 &
  \multicolumn{1}{c|}{79.26 ± 2.05} &
  87.17 ± 0.90 &
  93.57 ± 0.41 &
  93.58 ± 0.45 &
  67.22 ± 2.12 &
  81.46 ± 1.46 &
  62.93 ± 1.24 \\
DeepSets &
   &
  \second{75.49 ± 2.64 }&
  \second{58.31 ± 1.12 }&
  \multicolumn{1}{c|}{\second{54.52 ± 2.01}} &
  28.20 ± 39.89 &
  64.10 ± 19.94 &
  53.03 ± 27.79 &
  56.72 ± 5.95 &
  48.47 ± 15.91 &
  46.68 ± 5.65 \\
ProtoCount &
   &
  8.65 ± 2.83 &
  25.88 ± 0.83 &
  \multicolumn{1}{c|}{14.98 ± 2.48} &
  13.26 ± 8.39 &
  56.63 ± 4.19 &
  51.57 ± 10.17 &
  55.02 ± 12.35 &
  58.83 ± 6.48 &
  48.79 ± 8.02 \\
H2T &
   &
  35.20 ± 17.79 &
  35.05 ± 5.32 &
  \multicolumn{1}{c|}{30.69 ± 6.22} &
  77.50 ± 3.31 &
  88.65 ± 1.69 &
  88.69 ± 1.71 &
  51.28 ± 1.81 &
  49.82 ± 0.26 &
  51.03 ± 1.46 \\
OT &
   &
  39.79 ± 7.28 &
  37.73 ± 3.27 &
  \multicolumn{1}{c|}{31.92 ± 1.93} &
  \second{83.98 ± 0.89 }&
  \second{92.02 ± 0.41 }&
  \second{91.98 ± 0.46 }&
  63.16 ± 3.32 &
  62.86 ± 7.46 &
  \second{59.00 ± 2.49} \\
InfiniteGPFA &
   &
  0.62 ± 2.38 &
  27.00 ± 0.56 &
  \multicolumn{1}{c|}{13.28 ± 2.61} &
  4.50 ± 5.49 &
  52.26 ± 2.77 &
  45.92 ± 8.34 &
  51.67 ± 5.47 &
  57.03 ± 8.21 &
  46.34 ± 8.55 \\
PANTHER &
   &
  69.08 ± 3.63 &
  48.62 ± 5.28 &
  \multicolumn{1}{c|}{47.91 ± 5.45} &
  81.41 ± 2.38 &
  90.72 ± 1.14 &
  90.69 ± 1.21 &
  \second{63.00 ± 9.98 }&
  \second{69.18 ± 1.35 }&
  51.69 ± 10.68 \\ \cmidrule(r){1-1} \cmidrule(l){3-11} 
\textbf{NICER (Ours)} &
   &
  \best{88.92 ± 2.37} &
  \best{69.73 ± 3.02} &
  \multicolumn{1}{c|}{\best{69.56 ± 3.59}} &
  \best{88.47 ± 1.56} &
  \best{94.28 ± 0.76} &
  \best{94.23 ± 0.79 }&
  \best{68.06 ± 4.69 }&
  \best{76.99 ± 3.95 }&
  \best{62.14 ± 2.96 }\\ \bottomrule
\end{tabular}%
}
\end{table*}

\begin{figure*}[t]
    \centering
    \begin{minipage}[t]{0.56\textwidth} 
    \centering
    \captionof{table}{Performance of baselines on Morphological Prototyping tasks using a nonlinear transformer layer. The \best{best} and \second{second-best} results are highlighted in \best{bold red}, and \second{blue}, respectively. }
    \label{tab:morphology_prototype}
    \setlength\tabcolsep{4pt} 
    \resizebox{0.88\linewidth}{!}{%
        \begin{tabular}{@{}cl|cccc|cc@{}}
        \toprule
        \multicolumn{2}{c|}{\multirow{3}{*}{Method}}                  & \multicolumn{4}{c|}{Cancer Subtyping}              & \multicolumn{2}{c}{Survival Prediction} \\ \cmidrule(l){3-8} 
        \multicolumn{2}{c|}{}                    & \multicolumn{2}{c|}{PANDA}            & \multicolumn{2}{c|}{NSCLC} & LUAD    & BRCA    \\ \cmidrule(l){3-8} 
        \multicolumn{2}{c|}{}                    & Accur. & \multicolumn{1}{c|}{F1}    & Bal. Acc.      & F1        & C-Index & C-Index \\ \midrule
        \multicolumn{1}{c|}{\multirow{3}{*}{\rotatebox{90}{MIL}}}   & ABMIL    & 74.05 & \multicolumn{1}{c|}{74.42} & \second{94.19} & \second{94.23} & 62.20              & 77.35              \\
        \multicolumn{1}{c|}{} & DSMIL            & 72.48    & \multicolumn{1}{c|}{72.52} & \best{95.17}          & \best{95.19}     & \second{68.90}   & 77.72   \\
        \multicolumn{1}{c|}{} & ILRA             & \best{76.96}    & \multicolumn{1}{c|}{\second{77.11}} & 93.21          & 93.26     & 55.26   & \best{83.43}   \\ \midrule
        \multicolumn{1}{c|}{\multirow{7}{*}{\rotatebox{90}{Unsupervised}}} & DeepSets & 41.39 & \multicolumn{1}{c|}{37.12} & 89.51 & 90.38 & 65.18              & 61.88              \\
        \multicolumn{1}{c|}{} & ProtoCount       & 42.06    & \multicolumn{1}{c|}{35.41} & 58.38          & 51.61     & 57.19   & 48.99   \\
        \multicolumn{1}{c|}{} & H2T              & 61.07    & \multicolumn{1}{c|}{59.01} & 90.38          & 90.38     & 56.73   & 50.00   \\
        \multicolumn{1}{c|}{} & OT               & \second{68.48}    & \multicolumn{1}{c|}{67.21} & 91.35          & 91.34     & 64.22   & 65.38   \\
        \multicolumn{1}{c|}{} & InfiniteGPFA     & 26.85    & \multicolumn{1}{c|}{11.36}  & 50.00          & 34.41     & 50.00   & 50.37   \\
        \multicolumn{1}{c|}{} & PANTHER  & 70.02    & \multicolumn{1}{c|}{70.06} & 88.53          & 89.45     & 60.29   & 67.59   \\
        \multicolumn{1}{c|}{} & \textbf{NICER} & \best{76.96}    & \multicolumn{1}{c|}{\best{77.12}} & \best{95.17}          & \best{95.19}     & \best{70.81}   & \second{81.40}   \\ \bottomrule
        \end{tabular}%
    }
    \vspace{2.0cm}
        \begin{minipage}[t]{0.95\linewidth} 
            \centering
            \includegraphics[width=0.95\linewidth]{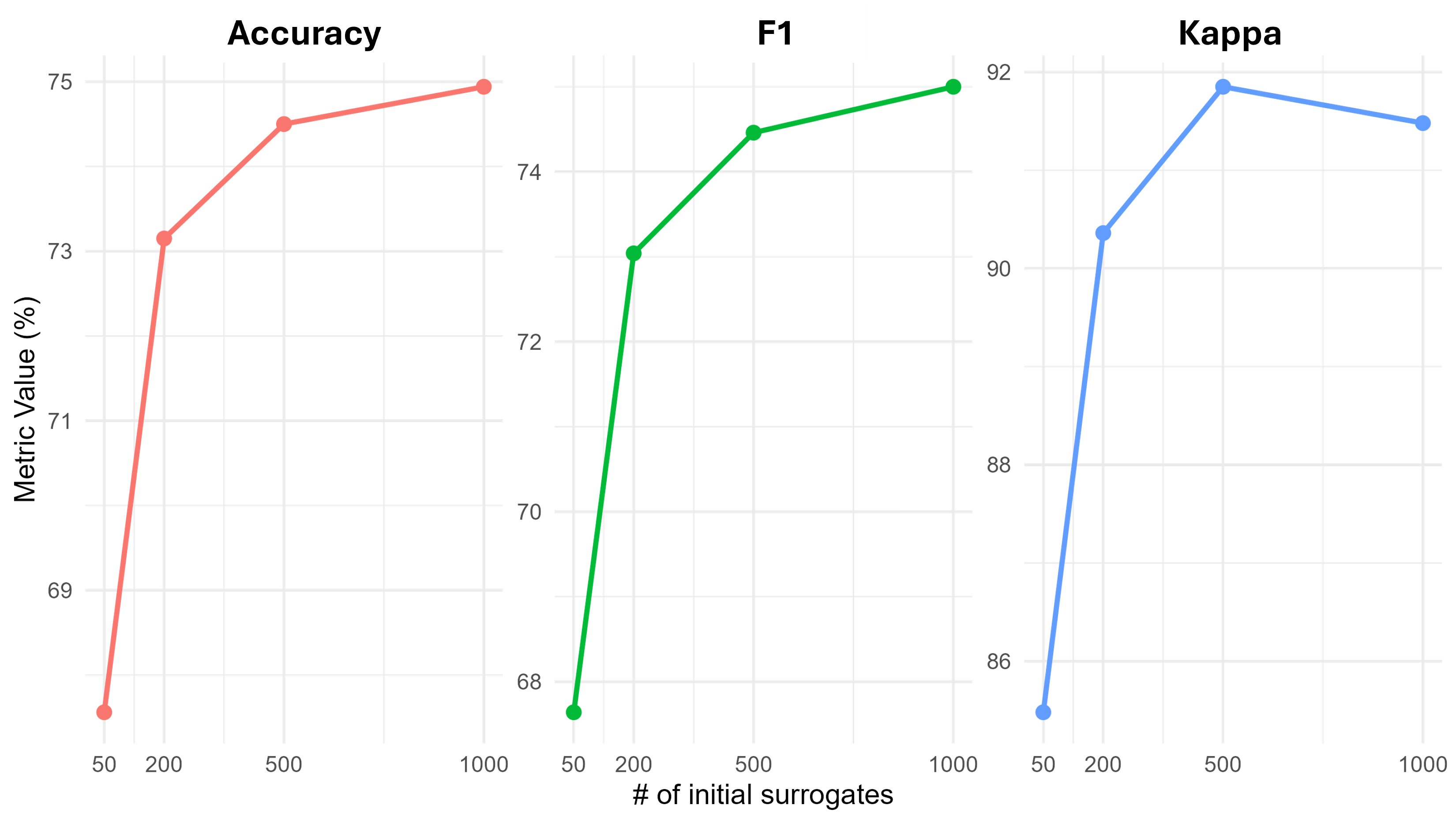}
            \captionof{figure}{Impact of initial patterns number on NICER's performance.}
            \label{fig:initial_P}
        \end{minipage}
\end{minipage}
    \hfill
    \begin{minipage}[t]
    {0.42\textwidth} 
    \vspace{0.0cm}
        \begin{minipage}[t]{\linewidth} 
            \centering
            \captionof{table}{Impact of top-$\kappa$ in pattern distillation on NICER's performance. Conducted on PANDA with ABMIL.}
            \label{tab:abl_top_k_perform}
            \resizebox{0.68\linewidth}{!}{%
            
                \begin{tabular}{@{}l|c|ccc@{}}
                \toprule
                Task                                              & \textbf{$\kappa$} & \textbf{Kappa} & \textbf{Acc.} & \textbf{F1} \\ \midrule
                \multirow{4}{*}{\rotcell{\begin{tabular}[c]{@{}l@{}}{\small Conden}\\ {\small -sation}\end{tabular}}}                    & 1                 & 90.18          & 70.02         & 69.87       \\
                                                                  & 3                 & 91.48          & 74.94         & 75.00       \\
                                                                  & 5                 & 90.47          & 72.93         & 72.72       \\
                                                                  & 10                & 89.52          & 70.47         & 70.49       \\ \midrule 
                \multicolumn{1}{c|}{\multirow{4}{*}{\rotcell{\small Prototyping}}} & 1                 & 92.89          & 75.17         & 75.47       \\
                \multicolumn{1}{c|}{}                             & 3                 & 93.07          & 77.85         & 77.96       \\
                \multicolumn{1}{c|}{}                             & 5                 & 92.77          & 75.62         & 75.80       \\
                \multicolumn{1}{c|}{}                             & 10                & 92.89          & 73.83         & 73.82       \\ \bottomrule
                \end{tabular}%
            } 
            \vspace{0.2cm}
        \end{minipage}
        \begin{minipage}[t]{\linewidth} 
            \centering
            \includegraphics[width=0.92\linewidth]{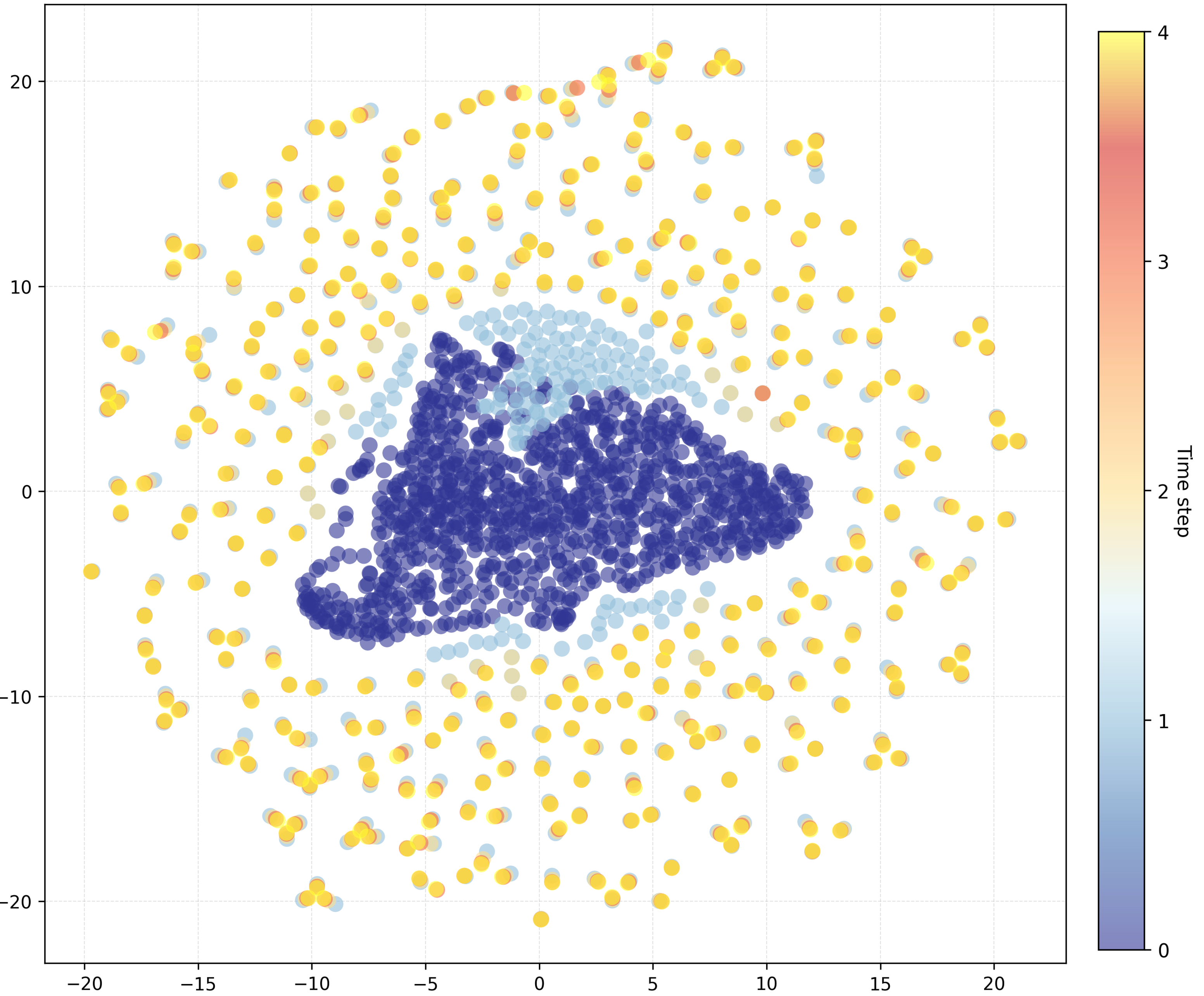}
            \captionof{figure}{NICER concepts disperse over iterations. This shows that our concept prototypes encode different information from the WSI, ensure diversity.}
            \label{fig:prototypes_tsne}
        \end{minipage}
        \end{minipage}
        \vspace{-2.0cm}
\end{figure*}

\subsection{Experiment Settings}    \label{sec:exp_settings}
\textbf{Datasets and Evaluation Metrics.}
For cancer subtyping, we evaluate NICER on two tasks: NSCLC subtyping on TCGA (2 classes), and ISUP grading on PANDA challenge (6 classes)~\citep{Bulten2022,Bulten2020-ju}. In survival prediction, we evaluate NICER with in-cohort TCGA-BRCA and TCGA-LUAD and cross-cohort evaluation with CPTAC. Following~\cite{song2024cvpr}, we evaluate the cancer subtyping tasks using Kappa~\citep{vieira10kappa}, weighted F1, accuracy for PANDA and balanced accuracy for NSCLC. For survival tasks, we report the concordance index (C-Index)~\citep{cindex}.


\textbf{Evaluation Settings.} We evaluate NICER and baselines along two axes: (i) condensation ability, by applying NICER and baselines on the training set and testing against the original WSI feature bag; and (ii) morphological prototyping, by applying methods to both training and test sets.

\textbf{Baselines.} 
We consider two categories of baselines: (1) {unsupervised prototyping} methods, which preserve task-agnostic representations, and (2) {MIL-based predictors}, which learn supervised slide-level representations from the whole bag. For unsupervised prototyping, we use DeepSets~\citep{zaheer17deepsets}, ProtoCounts~\citep{claudio2024}, H2T~\citep{Yu2023-bc}, InfiniteGPFA~\citep{yu2025infGPFA}, Optimal Transport (OT)~\citep{mialon2021a}, and PANTHER~\citep{song2024cvpr}. For MIL-based predictors, we adopt ABMIL~\citep{abmil}, DSMIL~\citep{li2021dsmil}, and ILRA~\citep{xiang2023exploring}. These models are either trained on condensed representations and evaluated on original feature bags to assess condensation quality, or used directly to evaluate NICER’s ability to produce unsupervised slide-level representations. Implementation details are provided in App.~\ref{app:implementation_details}.

\subsection{Performance on Condensation Ability} \label{sec:main_results}

\begin{wraptable}{r}{0.5\textwidth}
\centering
\caption{Computation cost comparison of NICER and baselines on NSCLC subtyping with ILRA}
\resizebox{0.92\linewidth}{!}{%
\begin{tabular}{@{}l|ccc@{}}
\toprule
\textbf{Metric}   & \textbf{Whole Bag} & \textbf{PANTHER} & \textbf{NICER} \\ \midrule
Storage (GB)      & 194                & 3.2              & 5.7            \\
Condense GPU (MB) & -                  & 5.30             & 9.17           \\
Condense Time (s) & -                  & 0.05             & 1.51           \\
Mean F1                & 93.58   & 90.69 & 94.23  \\ \bottomrule
\end{tabular}%
}
\label{tab:computation_cost}
\end{wraptable}

\textbf{Results on cancer subtyping.} 
Tab.~\ref{tab:main-results-dd} shows that NICER consistently outperforms competing methods across predictors for cancer subtyping. On PANDA, NICER achieves gains of 2.28-15.04\% with complex predictors (DSMIL, ILRA), and up to 33.72\% with simpler architectures (e.g., ABMIL). Notably, NICER incurs only a minor drop relative to Whole Bag ($\approx$2.87\% in Kappa), offering a favorable accuracy-efficiency trade-off. Similar trends are observed on NSCLC, where NICER improves balanced accuracy by up to 5.27\%, highlighting the effectiveness of its design.


\textbf{Results on survival prediction (In-cohort).}
In in-cohort survival prediction, NICER consistently outperforms competing methods across predictor architectures, with gains of up to 79.73\%. On BRCA, NICER improves the C-index by up to 33.40\% (4.98\% on average), while on NSCLC it achieves gains of 2.82\% and 4.70\% with ABMIL and DSMIL, respectively. With the ILRA predictor, NICER surpasses the Whole-Bag upper bound on several tasks, suggesting that condensation can denoise or regularize representations for survival prediction.



\textbf{Results on survival prediction (Cross-cohort).}
We evaluate the generalizability of NICER in cross-cohort settings against multiple baselines. As shown in Tab.~\ref{tab:main-results-dd}, NICER consistently outperforms all baselines across diverse predictor architectures by approximately 5\%. Moreover, NICER exhibits a substantially smaller performance gap between in-cohort evaluation on LUAD and cross-cohort evaluation on CPTAC compared to other methods, indicating that it learns transferable representations.

\subsection{Performance on Morphological Prototyping}
As shown in Tab.~\ref{tab:morphology_prototype}, NICER consistently outperforms other unsupervised methods across all scenarios, demonstrating its ability to preserve semantic information from the original noisy feature bag. Due to its nonparametric learning process and decomposition of diversity and compactness, NICER constructs representations that achieve competitive performance, exhibit only minimal drops across diverse tasks and datasets, underscoring its model-agnostic nature. In contrast, the performance of MIL methods is strongly dependent on the chosen architecture.


\begin{figure*}
    \centering
    \includegraphics[width=\linewidth]{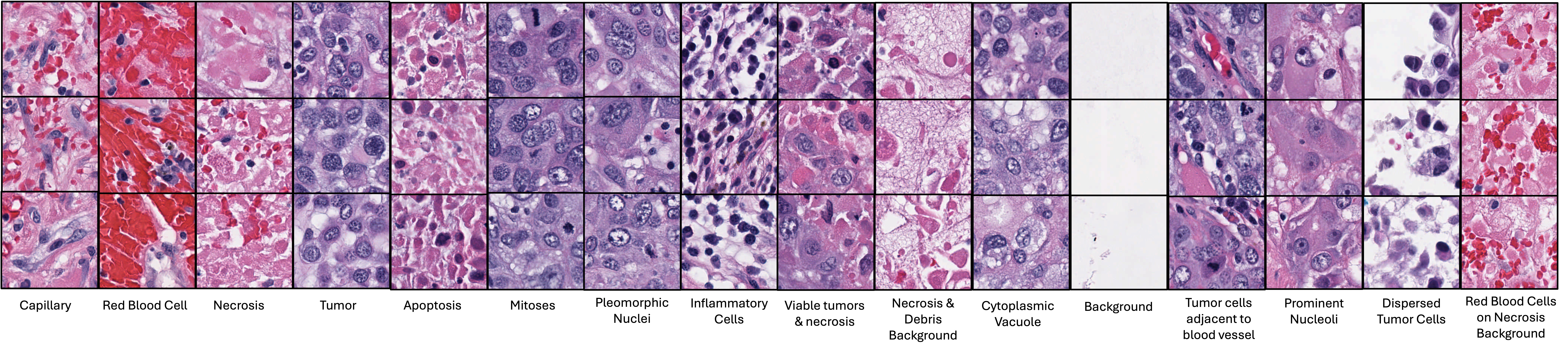}
    \caption{Morphological annotations of prototypes learned from NICER's learned global concepts by a board-certified pathologist on a WSI in the TCGA cohort (BRCA). Further details are in App.~\ref{app:biology}}
    \vspace{-0.5cm}
    \label{fig:biology}
\end{figure*}

\subsection{Additional Analysis \& Ablation Studies}

\textbf{Pathological Interpretation.}~Fig.~\ref{fig:biology} shows prototypical patterns extracted from a representative WSI in the TCGA-BRCA cohort and annotated by a board-certified pathologist. Each pattern exhibits high intra-pattern morphological consistency, indicating that the condensed representations preserve fine-grained cellular and tissue-level features. Moreover, NICER is able to separate closely related but diagnostically distinct patterns (e.g., viable tumor versus tumor-adjacent necrosis), which are critical for pathological assessment. Implementation details and additional analysis are in App.~\ref{app:biology}.

\textbf{Computation cost comparison. }~We compare the computational cost of NICER with the state-of-the-art PANTHER and the Whole Bag setting. As shown in Tab.~\ref{tab:computation_cost}, both PANTHER and NICER reduce storage from 194GB to 3.2–5.7GB. While NICER incurs higher per-iteration GPU and runtime costs during condensation to preserve fine-grained representations, this modest overhead leads to substantially better downstream performance. Because condensation is performed once per WSI, the additional cost is amortized and negligible relative to the gains. More details are discussed in App.~\ref{app:compute_complexity}.


\textbf{Impact of Initial Number of Patterns.}
Fig.~\ref{fig:initial_P} illustrates the effect of the number of patterns initialized at the start of the NICER algorithm on feature construction performance. Experiments are conducted on the PANDA dataset and evaluated on the Condensation Ability task using ABMIL as the predictor. As shown, NICER achieves substantial performance gains as the number of patterns increases, but the improvement begins to saturate beyond a certain point (e.g., $\sim$200). This trend arises because a larger pattern set provides greater distillation capacity, while the plateau reflects the condensation process, which removes redundancy and converges toward stable representations.

\textbf{Impact of Top-$\kappa$ in Pattern Learning.}
We conduct a sensitivity analysis on the number of patterns per patch (top-$\kappa$) for both Condensation Ability and Morphological Prototyping. As shown in Tab.~\ref{tab:abl_top_k_perform}, performance improves with increasing $\kappa$ but plateaus or slightly declines beyond $\kappa=3$, as overly large $\kappa$ diffuses patch information across patterns and weakens distillation, aligning with Fig.~\ref{fig:initial_P}.

\textbf{Concept Prototypes Diversity.}
To analyze the behavior of our condensation process, which maps patterns into concepts, we visualize 2D t-SNE embeddings of the concept set on the NSCLC dataset over the first five iterations. Fig.~\ref{fig:prototypes_tsne} shows that the learned concepts become increasingly dispersed (brighter points) as training progresses, reflecting convergence toward diverse and specific information. This demonstrates the effectiveness of the condensation process in reducing overlap and redundancy while preserving diversity, a key strength of NICER.
More results with different pretrained encoders and more extreme datasets can be found in App.~\ref{app:more_results}

\section{Related Work}
\textbf{Multiple Instance Learning.}
Early histology-based outcome prediction relied on pathologist-annotated regions of interest~\citep{Bychkov2018-nj,kather2019plos,nguyen2025sema}, while more recent work leverages WSIs for clinical prediction using multiple instance learning (MIL)~\citep{Campanella2019-wp,chen22cvpr,Tang_2023_ICCV,nguyen2025bridging, li2021dual,Lu_2023_CVPR,shao2021transmil,Tang_2023_ICCV,xiang2023exploring}. MIL represents each WSI as an unordered set of patches, effectively treating each slide as a dataset and naturally motivating condensation. Thus, NICER is formulated as a MIL-agnostic condensation method prioritizing computational and storage efficiency.


\textbf{Dataset Condensation. }~Dataset condensation synthesizes compact datasets for training neural networks, reducing storage and training costs~\citep{wang2020dd,nguyen2021datameta,nguyen2021dataset}. The most common approach is data-matching methods which optimize synthetic data to mimic the learning behavior of real data. For example, \cite{zhao2021dataset} match gradients between real and synthetic data, while trajectory-matching methods~\citep{cazenavette2022distillation,du2023minimizingaccumulatedtrajectoryerror,DonLe_HighOrder_MICCAI2025} extend gradient matching by aligning multi-step training dynamics. Factorization-based approaches~\citep{deng2022remember,zheng2024leveraginghierarchicalfeaturesharing} distill data into low-dimensional manifolds for instance recovery. However, most existing condensation methods rely on expensive supervision, limiting their applicability to histological imaging.

\textbf{Self-supervised Data Condensation and Computation Pathology Applications.} 
Self-supervised data condensation has recently been proposed to condense unlabeled datasets while preserving task-agnostic information transferable to downstream tasks~\citep{lee2024selfsupervised,yu2025boostselfsuperviseddatasetdistillation}. In computational pathology, condensation is particularly well motivated, as recurring histological patterns often correspond to shared tissue morphology~\citep{hou2016cvpr,Kalra2020-fc,pan2022human,wang2022,yan2012cvpr,Yu2023-bc}. Nevertheless, existing methods typically condense WSIs into prototypes using handcrafted features~\citep{claudio2024,zaheer17deepsets} or unsupervised clustering heuristics~\citep{vu2023h2t,song2024cvpr}. This calls for a new perspective that formulates WSI condensation in a principled manner, enabling theoretically grounded solutions aligned with the task-agnostic representation space. 

\section{Conclusions and Limitations}

\textbf{Conclusion.}~We reformulated whole-slide image condensation as a distribution-matching problem under a fixed representation lens and showed that it admits a principled nonparametric solution via Poisson point process–based MAP inference. This view naturally allows the number of condensed prototypes to adapt to the complexity of each slide and motivates NICER, a tractable learning-based framework for per-slide condensation. Experiments across multiple histopathology benchmarks, supported by pathological evaluation, demonstrate consistent improvements in efficiency-accuracy trade-offs, underscoring the benefits of a theoretically grounded approach to scalable WSI representation learning. 

\textbf{Limitations.}~NICER has several limitations. First, it adopts a bag-of-features formulation and does not explicitly model spatial or multi-scale dependencies among patches, which may restrict its ability to capture fine-grained tissue organization in tasks where spatial context is important. Second, the PPP prior assumes independent concept occurrence, which may be insufficient for representing strongly correlated or highly complex morphological patterns. Finally, NICER operates on fixed pretrained representations. Although we evaluate two foundation models, a broader study across diverse encoders is needed to more fully characterize its robustness and generalizability across settings.
We acknowledge such limitations in details in App.~\ref{app:limitations}.

\begin{ack}
This research was funded by the Vietnam National Foundation for Science and Technology Development (NAFOSTED) under project code 24/2025/HĐ-TN.01; and supported in part by the VinUni-Illinois Smart Health Center. In addition, Trong Nghia Hoang is supported by the National Science Foundation CAREER Award IIS-2544071.
\end{ack}


\medskip

\bibliographystyle{plain}
\bibliography{refs}
\newpage
\appendix

\textbf{\Large{Appendix: Nonparametric Distribution Matching
for Self-Supervised Whole-Slide Image Condensation}}
\startcontents[appendix]
\printcontents[appendix]{}{1}{}

\newpage


\section{Broader Statement of Impact} \label{app:impact}

This research develops an effective nonparametric compression and condensation algorithm for whole-slide images (WSIs), enabling efficient learning from histological data with varying complexity. The mathematical methods and insights presented in this work help bridge the gap between large, gigabyte-scale images and their practical applications in healthcare. While the potential application of our methods to real patient data may raise ethical considerations, such effects are indirect and not the focus of this study. Our experiments rely solely on publicly available datasets, ensuring that no ethical concerns are introduced in the evaluation of our algorithms.

\section{Pseudocode for NICER} \label{app:pseudocode}

\begin{algorithm}[h]
\caption{Nonparametric Self-Supervised Data Condensation (NICER)}
\label{alg:NICER}
\textbf{Input:} WSI feature bag $\mathcal{H}=\{h_i\}_{i=1}^{N}$, iterations $T$, initial patterns $M$, top-$\kappa$ \\
\textbf{Output:} condensed concept prototype set $\Phi$
\begin{algorithmic}[1]
\STATE Initialize pattern set $\mathcal{Z}$ with $M$ patterns
\FOR{$t=1$ \textbf{to} $T$}
    \FOR{$i=1$ \textbf{to} $N$}
        \STATE Query top-$\kappa$ relevant patterns from $\mathcal{Z}$ for $h_i$ using Eq.~\ref{eq:distill_sim}
        \STATE Distill $h_i$ into the selected patterns by maximizing Eq.~\ref{eq:distill_sim}
    \ENDFOR
    \STATE $(\Phi,b) \leftarrow \textsc{condense}(\mathcal{Z})$ \hfill (solve Eq.~\ref{eq:psi_final})
    \STATE Remove unused concepts: $\Phi \leftarrow \{\phi_k \in \Phi \mid \sum_{m=1}^{M} b_{mk} > 0\}$
    \STATE Reinitialize patterns: $\mathcal{Z} \leftarrow \Phi$
\ENDFOR
\STATE \textbf{return} $\Phi$
\end{algorithmic}
\end{algorithm}

\begin{algorithm}[h]
\caption{\textsc{condense}: Alternating Optimization for Prototype Condensation}
\label{alg:condense}
\textbf{Input:} pattern set $\mathcal{Z}=\{z_m\}_{m=1}^{M}$, inner iterations $S$ \\
\textbf{Output:} concept set $\Phi$ and assignments $b$
\begin{algorithmic}[1]
\STATE Initialize concept set $\Phi=\{\phi_k\}_{k=1}^{K}$ from $\mathcal{Z}$ with $K \equiv M$
\STATE Initialize one-hot assignments $b \in \{0,1\}^{M\times K}$ with $\sum_{k=1}^{K} b_{mk}=1$
\STATE Initialize parameters $\gamma$ and $\zeta$
\FOR{$s=1$ \textbf{to} $S$}
    \STATE Update $\gamma$ given $\Phi$ and $b$ \hfill (Eq.~\ref{eq:N_new})
    \STATE Update $\zeta$ given $b$ \hfill (Eq.~\ref{eq:cate_b})
    \STATE Update $\Phi$ given $\gamma$ and $b$ \hfill (Eq.~\ref{eq:N_new})
    \FOR{$m=1$ \textbf{to} $M$}
        \STATE Compute scores $\{s_{mk}\}_{k=1}^{K}$ using Eq.~\ref{eq:final_linear_form}
        \STATE Assign $b_m \leftarrow \mathrm{onehot}\!\left(\arg\max_{k} s_{mk}\right)$
    \ENDFOR
    \STATE Remove unused concepts from $\Phi$ using $b$
\ENDFOR
\STATE \textbf{return} $(\Phi,b)$
\end{algorithmic}
\end{algorithm}

\section{Extended Related Work} \label{app:related_work}

\textbf{Multiple Instance Learning.}
Early histology-based outcome prediction relied on pathologist-annotated regions of interest~\citep{Bychkov2018-nj,kather2019plos,nguyen2025sema}, while more recent work leverages WSIs for clinical prediction using multiple instance learning (MIL)~\citep{Campanella2019-wp,chen22cvpr,Tang_2023_ICCV,nguyen2025bridging}. MIL represents each WSI as a set of patches or data instances, effectively treating each slide as a separate dataset and naturally motivating dataset condensation. Considerable effort has focused on improving MIL through better patch aggregation strategies~\citep{li2021dual,Lu_2023_CVPR,shao2021transmil,Tang_2023_ICCV,xiang2023exploring}, interpretability~\citep{javed2022addmil,kevin2022zooming}, and efficiency~\citep{xiang2023exploring,qiu2023sc}. We formulate NICER as a self-supervised data condensation method with primary priority on computing and storage efficiency, which therefore operates across different MIL frameworks.

\textbf{Dataset Condensation. }~Dataset condensation synthesizes compact datasets for training neural networks, reducing storage and training costs~\citep{wang2020dd,nguyen2021datameta,nguyen2021dataset}. Traditional coreset selection instead chooses representative samples from the original dataset~\cite{guo2022learning,kim2022differentiable,lee19d,mialon2021a}, but its non-editable nature limits performance. In contrast, data-matching methods optimize synthetic data to mimic the learning behavior of real data. For example, \cite{zhao2021dataset} match gradients between real and synthetic data, inspiring a class of bilevel optimization approaches. Trajectory-matching methods~\citep{cazenavette2022distillation,du2023minimizingaccumulatedtrajectoryerror,DonLe_HighOrder_MICCAI2025} extend gradient matching by aligning multi-step training dynamics, while factorization-based approaches~\citep{deng2022remember,zheng2024leveraginghierarchicalfeaturesharing} distill data into low-dimensional manifolds for instance recovery. However, most existing condensation methods rely on supervision and costly full training, limiting their applicability to histological imaging.

\textbf{Self-supervised Data Condensation and Computation Pathology Applications.} Motivated by the rise of pretrained foundation models, self-supervised data condensation (SSDC) is recently proposed to perform data condensation with unlabeled images, aiming to preserve important task-agnostic information which can be transferred to downstream tasks~\cite{lee2024selfsupervised,yu2025boostselfsuperviseddatasetdistillation}
In computational pathology, self-supervised data condensation is a natural and necessary problem since repeating histology patterns often reflect shared morphology~\citep{hou2016cvpr,Kalra2020-fc,pan2022human,wang2022,yan2012cvpr,Yu2023-bc}. 
However, it is impractical to adopt such modern data condensation methods in histological context due to the enormous size of each WSI leading to expensive condensation cost for the whole dataset during distillation. For this reason, most approaches condense WSI into prototypes by using manual features~\citep{claudio2024,zaheer17deepsets} or unsupervised clustering~\citep{vu2023h2t,song2024cvpr}.

\section{Implementation Details} \label{app:implementation_details}
\subsection{Datasets}
We provide brief explanations for the datasets that were used for the evaluation of NICER for condensation ability and prototyping ability.

\textbf{PANDA.}~\citep{Bulten2022,Bulten2020-ju} For the ISUP grading task, we used prostate cancer core needle biopses (n=10,616) from the Prostate Cancer Grade Assessment (PANDA) challenge. Each biopsy is given an ISUP grade, making this a 6-class classification task. These biopses are collected from Karolinska Institute (KRLS) and Radboud University Medical Center (RUMC). We label-stratify the PANDA dataset into train/val/test of 80:10:10 and performance was evaluted using Cohen's quadratic weighted Kappa $\kappa^2$, accuracy and weighted F1 metrics.

\textbf{NSCLC.} For the non-small cell lung carcinoma (NSCLC) subtyping task, we use H\&E WSIs from TCGA for classifying lung adenocarcinoma (LUAD) and lung squamous cell carcinoma (LUSC) cases. The TCGA cohort contains a total of 1,041 slides (LUAD: 529, LUSC: 512). We label-stratify the TCGA cohort into train/val/test fold of 80:10:10 and evaluate performance using Cohen's Kappa, balanced accuracy and F1 metrics.

\textbf{TCGA-BRCA.} The Breast Invasive Carcinoma (BRCA) cohort from The Cancer Genome Atlas (TCGA). a joint effort of the NCI and NHGRI, provides one of the most comprehensive digital pathology resources for breast cancer. The dataset contains 1,133 diagnostic WSIs spanning 1,062 patients, covering diverse histological and molecular subtypes. In addition to imaging, the cohort supplies curated clinical outcomes, including Overall Survival (OS) and Progression-Free Interval (PFI), which have been widely adopted as endpoints for developing and benchmarking survival prediction models. We label-stratify the TCGA-BRCA dataset into train/val/test fold of 80:10:10 and evaluate performance C-Index for survival prediction

\textbf{TCGA-LUAD.} The Lung Adenocarcinoma (LUAD) cohort within TCGA offers a large-scale, multi-institutional collection of pathology images and outcome data for a major subtype of non-small cell lung cancer. It includes 529 diagnostic WSIs corresponding to 478 patients, with extensive clinical annotation. As with BRCA, the LUAD cohort provides OS and PFI as standardized survival endpoints, enabling robust prognostic modeling and cross-study comparison in computational pathology research. We label-stratify the TCGA-LUAD dataset into train/val/test fold of 80:10:10 and evaluate performance C-Index for survival prediction

\textbf{CPTAC.}~We use CPTAC as an external evaluation cohort to assess cross-cohort generalization on the survival prediction task. Models trained on the TCGA cohort are directly evaluated on CPTAC LUAD (185 patients, 486 WSIs) without further adaptation. This setting is substantially more challenging due to pronounced dataset shifts, including variations in image acquisition and H\&E staining protocols, differences in geographic and demographic distributions (reflecting social determinants of health), and potential cohort-specific biases such as discrepancies in clinical follow-up and outcome annotation procedures between TCGA and CPTAC.

\subsection{Evaluation Setting Details}
We view unsupervised prototype construction not only as a condensation framework, whose primary role is to reduce redundancy in the WSI feature bag, but also as a form of representation learning, since the resulting prototypes are later used as inputs for downstream predictors. Accordingly, we evaluate all methods along two complementary axes:

\textbf{Condensation ability.}
Here, the goal is to assess how well the condensed prototype set preserves information from the original WSI feature bag. Specifically, we apply NICER and baseline methods on the training set to obtain prototypes, and then evaluate them on the original uncompressed bag of features. This setup isolates the effectiveness of condensation by measuring how much predictive power is retained (or lost) after summarization, independent of downstream task complexity. It answers the key question: \textit{Does condensation discard critical information or faithfully represent the original slide?}

\textbf{Morphological prototyping.}
In this setting, we follow the prior evaluation protocol~\cite{song2024cvpr} to test whether 
condensed prototypes can generalize as useful, task-agnostic representations. Condensation 
methods are applied to both the training and test sets, and the resulting prototypes are directly 
used for downstream prediction. Unlike the condensation ability evaluation, this setup emphasizes the representation learning capacity of the prototypes, focusing on whether they capture robust morphological cues transferable across unseen WSIs rather than reconstruction fidelity. 
To ensure fairness, we use a nonlinear transformer 
layer for all baselines since NICER produces prototype sets whose capacity adapts to 
WSI complexity and thus cannot be fully exploited by a fixed MLP layer.

We report the standard deviations for Tab.~\ref{tab:main-results-dd} by running with different random seeds.

\subsection{Baseline Details.}
We consider two baseline categories: 1) \textit{unsupervised prototyping} methods, which collapse hierarchy in condensation, followed by a task-specific neural predictor, and 2) \textit{MIL-based predictors}, which construct supervised slide-level representations from the whole bag. For \textit{unsupervised prototyping}, we compare {NICER} with {DeepSets}~\citep{zaheer17deepsets}, {ProtoCounts}~\citep{claudio2024}, {H2T}~\citep{Yu2023-bc}, {InfiniteGPFA}~\citep{yu2025infGPFA}, {Optimal Transport (OT)}~\citep{mialon2021a}, and {PANTHER}~\citep{song2024cvpr}. Specifically, DeepSets, ProtoCounts, and H2T build prototypes from distance-based clustering (e.g., K-Means); InfiniteGPFA is adapted to perform nonparametric factor analysis on each WSI; and OT and PANTHER allow soft prototypical assignment. For \textit{MIL-based predictors}, we adopt three supervised baselines: attention-based MIL (ABMIL)~\citep{abmil}, dual-stream MIL (DSMIL)~\citep{li2021dsmil}, and low-rank MIL (ILRA)~\citep{xiang2023exploring}. These models are trained on unsupervised prototypes and evaluated on original feature bags to measure condensation ability, or used directly to assess NICER’s effectiveness in producing slide-level unsupervised representations.

\subsection{Training Details and Computational Resources}

\textbf{Data preparation. }
WSIs at $20\times$ magnification ($0.5,\mu m$/pixel) are divided into non-overlapping $256 \times 256$ patches, and all patches are used without sampling. These patches are converted to representations using UNI~\cite{chen2024uni}, a pretrained foundation encoder. We set $\kappa=3$ and found $T=20$ iterations sufficient for convergence across all datasets.

\textbf{Hyperparameter settings. } 
For training, we adopt the AdamW optimizer with weight decay set to $1 \times 10^{-5}$ and employ a cosine decay learning rate scheduler. In the \textit{cancer subtyping} experiments, models are trained with cross-entropy loss for up to 50 epochs, with early stopping triggered if the validation loss fails to improve for 10 consecutive epochs. The initial learning rate is set to $1 \times 10^{-4}$. Since both the original feature bags and the NICER's representations form variable-length WSI sets, we use a batch size of 1 combined with gradient accumulation over 32 steps across all methods. For the \textit{survival prediction} task, we optimize using the negative log-likelihood (NLL) loss~\citep{zadeh21nll}, training over 50 epochs with a per-patient batch size of 1 and an initial learning rate of $1 \times 10^{-5}$. The training of NICER further involves an alternating optimization process over 20 iterations, with $\kappa = 3$ and an initial pattern set size of $M_0 = 1000$, which decreases progressively through the condensation procedure. As a nonparametric random process, NICER allows us to directly regulate the number of generated prototypes via $\kappa$ and $M_0$. Following prior practices~\citep{vu2023h2t,song2024cvpr}, we set a number of prototypes generated by unsupervised baselines to 16 for all WSIs used in our experiments. The implementation details of predictor architectures, unsupervised baselines follow original papers and previous settings~\citep{song2024cvpr}.

\textbf{Computational considerations.} All experiments and data preprocessing are conducted on a NVIDIA RTX A6000 with 46GB of memory.

\section{Space Complexity and Scalability} \label{app:compute_complexity}
In this section, we further analyse the space complexity and the practical scalability of NICER for real WSIs. The main computational and memory bottleneck in NICER is Local Pattern Exploration, which computes similarities between $N$ patch embeddings and the initial pattern set $\mathcal Z$ of size $M$. The dominant memory term is the patch-pattern similarity matrix in $\mathbb{R}^{N \times M}$, together with the patch embeddings and initial patterns, giving overall space complexity
\[
\mathcal O(Nd + Md + NM).
\]
For example, on an NVIDIA RTX A6000 with 48 GB VRAM, using $M=1000$, $d=1024$, and float32 precision, the maximum processable number of patches is approximately
\[
\frac{48\times 10^9 - 1000\cdot 1024\cdot 4}{(1024+1000)\cdot 4}
\approx 5.9\times 10^6.
\]
This comfortably exceeds the scale of typical clinical WSIs at 0.5 MPP, where even extreme cases usually contain fewer than $10^6$ tissue patches after background filtering. Larger WSIs can be handled by batching over the $N$ dimension, since patch-pattern similarity computations are independent and batching does not alter the optimization objective.

\section{Nonparametric Prior in NICER}
\label{app:nicer_nonparametric_prior}

In this section, we clarify how the nonparametric prior in NICER is encoded through the active-set optimization procedure. The key point is that NICER does not fix the number of effective concepts in advance. Instead, the effective concept set is determined by the assignments from local patterns to global concepts, and concepts receiving no assignments are removed. This implements a cardinality penalty corresponding to the nonparametric prior.

We begin from the full MAP objective over the local pattern set $\mathcal Z_w$ and the slide-specific concept configuration $\Phi_w$:
\begin{align}
    \max_{\mathcal Z_w, \Phi_w}
    \log \mathbb P(\mathcal H_w, \mathcal Z_w, \Phi_w)
    =
    \max_{\mathcal Z_w, \Phi_w}
    \bigg\{
        \log \mathbb P(\mathcal H_w \mid \mathcal Z_w)
        + \log \mathbb P(\mathcal Z_w \mid \Phi_w)
        + \log \mathbb P(\Phi_w)
    \bigg\}.
    \label{eq:full_map_nicer}
\end{align}
Here, $\mathcal H_w$ denotes the empirical patch-embedding bag of slide $w$, $\mathcal Z_w$ denotes the learned local patterns, and $\Phi_w$ denotes the latent concept configuration. The first term encourages $\mathcal Z_w$ to preserve the empirical patch distribution, the second term assigns local patterns to latent concepts, and the final term is the nonparametric prior over the concept configuration.

For a fixed slide $w$, we drop the slide index for simplicity. Using the likelihood and assignment model defined in Eqs.~\ref{eq:N_old}, \ref{eq:N_new}, \ref{eq:assignment_prior}, and \ref{eq:psi_final}, the condensation component of Eq.~\eqref{eq:full_map_nicer} can be written as
\begin{align}
    \max_{\mathcal Z, \Phi, b, \gamma, \zeta}
    \bigg\{
        \sum_m \log \mathbb P(z_m \mid b_m, \Phi, \gamma)
        +
        \sum_m \log \mathbb P(b_m \mid \zeta)
        +
        \log \mathbb P(\Phi)
    \bigg\}.
    \label{eq:map_with_prior}
\end{align}
Under the Poisson point process prior used in our formulation (see Theorem~\ref{thm:ppp_map}), the log-prior over a finite concept configuration contributes a penalty proportional to its cardinality. Up to constants independent of the optimization variables, we therefore obtain
\begin{align}
    \max_{\mathcal Z, \Phi, b, \gamma, \zeta}
    \bigg\{
        \sum_m \log \mathbb P(z_m \mid b_m, \Phi, \gamma)
        +
        \sum_m \log \mathbb P(b_m \mid \zeta)
        -
        \lambda |\Phi|
    \bigg\},
    \label{eq:map_cardinality_penalty}
\end{align}
where $\lambda > 0$ is determined by the intensity of the point-process prior and the inverse temperature of the corresponding Gibbs posterior. Thus, the nonparametric prior favors compact concept configurations while still allowing the number of concepts to adapt to the complexity of each slide.

Importantly, $|\Phi|$ is not treated as a fixed constant in NICER. Although the algorithm may be initialized with a large candidate concept set, only concepts that are assigned to at least one local pattern remain active. We therefore define the active concept set as
\begin{align}
    \Phi^{+}
    =
    \left\{
        \phi_k \in \Phi :
        \sum_m b_{mk} > 0
    \right\},
    \qquad
    |\Phi^{+}|
    =
    \left|
    \left\{
        k :
        \sum_m b_{mk} > 0
    \right\}
    \right|.
    \label{eq:active_set_definition}
\end{align}
With this active-set parameterization, the effective MAP objective optimized by NICER becomes
\begin{align}
    \max_{\mathcal Z, \Phi, b, \gamma, \zeta}
    \bigg\{
        \sum_m \log \mathbb P(z_m \mid b_m, \Phi, \gamma)
        +
        \sum_m \log \mathbb P(b_m \mid \zeta)
        -
        \lambda |\Phi^{+}|
    \bigg\}
    \label{eq:active_set_objective}
\end{align}
or equivalently,
\begin{align}
    \max_{\mathcal Z, \Phi, b, \gamma, \zeta}
    \bigg\{
        \sum_m \log \mathbb P(z_m \mid b_m, \Phi, \gamma)
        +
        \sum_m \log \mathbb P(b_m \mid \zeta)
        -
        \lambda
        \left|
        \left\{
            k :
            \sum_m b_{mk} > 0
        \right\}
        \right|
    \bigg\}.
    \label{eq:active_assignment_objective}
\end{align}

This form makes the role of the nonparametric prior explicit. The assignment variables $b$ determine which concepts are used, and the cardinality term penalizes the number of active concepts. During the alternating optimization described in App.~\ref{app:pseudocode}, concepts with no assigned patterns are removed from the active set. This pruning step is therefore not a heuristic post-processing operation; it is the active-set update corresponding to the PPP-induced cardinality penalty in the MAP objective. In this sense, NICER performs approximate nonparametric MAP inference over a reduced support: the local pattern set $\mathcal Z$ approximates the empirical patch distribution, while the active concept set $\Phi^{+}$ approximates the variable-cardinality prototype configuration favored by the nonparametric prior.

We provide the connection between theoretical objects in our analysis and NICER approximation components in Tab.~\ref{app:tab:connection_theo_algo}.

\begin{table}[t]
\centering
\caption{Connection between the theoretical objective and NICER.}
\label{app:tab:connection_theo_algo}
\begin{tabular}{lll}
\toprule
Theoretical object & Role & NICER approximation \\
\midrule
$\mu_w$ & Empirical slide distribution & Reduced support $\mathcal Z_w$ \\
$OT(\mu_w,\nu_{\Phi_w})$ & Distribution matching & Gaussian pattern-concept likelihood \\
$\lambda|\Phi_w|$ & Cardinality penalty & Active concept set $|\Phi_w^+|$ \\
PPP prior $\Pi_\rho(\Phi_w)$ & Variable-capacity prior & Pruning unassigned concepts \\
MAP inference & Prototype selection & Alternating updates over $(\Phi,\gamma,\zeta)$ and $b$ \\
\bottomrule
\end{tabular}
\end{table}

\section{Derivation of $\mathbb{P}(\mathcal{H} \mid \mathcal{Z})$}
\label{app:derive_HZ}

We begin from the conditional likelihood as follow:
\begin{equation}
\log \mathbb{P}(\mathcal{H} \mid \mathcal{Z}, \theta)
= \sum_{i=1}^N \log \mathcal{N}(h_i \mid z_{(i)}^\ast, \sigma^2 \mathbf{I}),
\end{equation}
where each feature $h_i$ is modeled as a Gaussian centered at its assigned 
pattern prototype $z_{(i)}^\ast$. Expanding the Gaussian log-likelihood yields:
\begin{equation}
\log \mathbb{P}(\mathcal{H} \mid \mathcal{Z}, \theta)
= -\frac{1}{2\sigma^2} \sum_{i=1}^N \|h_i - z_{(i)}^\ast\|^2 + C,
\end{equation}
with $C$ denoting terms independent of $h_i$ or $z_{(i)}^\ast$. The squared 
distance can be written as
\begin{equation}
\|h_i - z_{(i)}^\ast\|^2 = \|h_i\|^2 + \|z_{(i)}^\ast\|^2 
- 2\langle h_i, z_{(i)}^\ast\rangle.
\end{equation}
Since both $h_i$ and $z_{(i)}^\ast$ are $\ell_2$-normalized embeddings, {we 
have $\|h_i\|^2 = \|z_{(i)}^\ast\|^2 = 1$}. This reduces the 
squared distance to
\begin{equation}
\|h_i - z_{(i)}^\ast\|^2 = 2 - 2\langle h_i, z_{(i)}^\ast\rangle.
\end{equation}
Plugging this into the log-likelihood gives
\begin{equation} \label{eq:distill_sim}
    \log \mathbb{P}(\mathcal{H} \mid \mathcal{Z}, \theta)
\approx -\frac{1}{2\sigma^2} \sum_{i=1}^N \big(2 - 2\langle h_i, z_{(i)}^\ast\rangle\big) + C.
\end{equation}
This shows 
that maximizing the Gaussian likelihood is approximately equivalent to 
maximizing feature - prototype similarity, providing a probabilistic justification 
for our design.

\textbf{Design and Learning.}
We implement this objective as a retrieval-based pattern selection process using cosine similarity $
\Gamma(h_i, z_m) \triangleq \langle h_i, z_m \rangle $.
To allow flexibility while maintaining compactness, each feature distributes its mass over its top-$\kappa$ nearest patterns ($\kappa \ll M$). This yields a stable and efficient approximation of $\mathcal{H}$ by $\mathcal{Z}$, which we treat as an intermediate latent representation. 

\textbf{Discussion on Top-$\kappa$ Approximation}.~As shown in App.~\ref{app:pseudocode}, the Top-$\kappa$ step is positioned where revisitation occurs after a single iteration of nonparametric compression. In this context, the algorithm iterates over local patches once more and distills them into the current shared prototypes. However, distilling into irrelevant prototypes collapses the information and produces isolated local prototypes that are semantically unrelated to any others, which are subsequently deactivated during the next iteration of concept condensation. Such distillation therefore incurs extra computational cost without improving performance, which is why we apply top-$\kappa$ assignment to avoid this. This is confirmed by our experiments in Tab.~\ref{tab:abl_top_k_perform}, where performance improves from $\kappa=1$ to $\kappa=3$ and plateaus beyond $\kappa=5$, indicating that the residual approximation gap becomes negligible once patterns are well-separated.

From a theoretical point of view, we anticipate that this will not lead to any performance degradation based on the conceptual and empirical validation above. However, formalizing this into a rigorous theoretical analysis is an interesting future work, which fortunately does not affect the key contribution in the current work.

\section{Proof of Theorem~\ref{thm:ppp_map}}
\label{app:thm_proof}

\subsection{Problem Setup}

Following the problem formulation derived in Sec.~\ref{sec:prelim}, our objective is to match these empirical measures under a fixed representational viewpoint using optimal transport:
\begin{equation}
\min_\theta\;
\mathrm{OT}\!\left(\mu_w,\;\nu_{\Phi_w, \theta}\right),
\label{eq:globalOT}
\end{equation}
where $\mathrm{OT}$ denotes optimal transport with squared-Euclidean cost.

\textbf{Prototype-based condensation.}
Condensation is required to be both compact and reusable across slides. Following the problem formulation derived in Sec.~\ref{sec:prelim}, the per-slide distribution matching problem becomes
\begin{equation}
\min_{\Phi_w,\theta}\;
\mathrm{OT}\!\left(\mu_w,\;\nu_{\Phi_w,\theta}\right).
\label{eq:globalOTPhi}
\end{equation}
Since the prototype cardinality $|\Phi_w|$ is not determined by distribution matching alone,
we introduce a complexity penalty that controls compression:
\begin{equation}
\min_{\Phi_w,\theta}\;
\mathrm{OT}\!\left(\mu_w,\;\nu_{\Phi_w,\theta}\right)
+\lambda|\Phi_w|,
\qquad \lambda>0.
\label{eq:penalizedW2}
\end{equation}
We establish that this objective is equivalent to MAP inference over $(\Phi_w, \theta)$ under a Poisson point process prior, a correspondence that both justifies the cardinality penalty probabilistically and motivates our alternating MAP-based optimization algorithm.

\subsection{Poisson Point Process prior and MAP equivalence}

\textbf{PPP prior and induced cardinality penalty.}
Let \(\Pi_\rho\) be a homogeneous Poisson point process (PPP) on a bounded domain
\(\Omega \subset \mathbb{R}^d\) with constant intensity \(\rho\in(0,1)\). A realization is a finite
configuration
\begin{equation}
\Phi_w=\{\phi_1,\ldots,\phi_n\}\subset \Omega,
\qquad n=|\Phi_w|.
\end{equation}
By the standard characterization of a homogeneous PPP,
\begin{equation}
|\Phi_w|\sim \mathrm{Poisson}(\rho|\Omega|),
\end{equation}
and conditional on \(|\Phi_w|=n\), the points
\(\phi_1,\ldots,\phi_n\) are i.i.d. uniform on \(\Omega\). Therefore, the \(n\)-point density of
\(\Pi_\rho\) is
\begin{equation}
j_{\rho,n}(\phi_1,\ldots,\phi_n)
=
\exp(-\rho|\Omega|)\rho^n \frac{1}{n!|\Omega|^n}.
\end{equation}

Let \(\Pi_1\) denote the unit-rate PPP on the same domain. Its \(n\)-point density is
\begin{equation}
j_{1,n}(\phi_1,\ldots,\phi_n)
=
\exp(-|\Omega|) \frac{1}{n!|\Omega|^n}.
\end{equation}
For a given configuration \(\Phi_w=\{\phi_1,\ldots,\phi_n\}\), with \(n=|\Phi_w|\), the
Radon--Nikodym derivative of \(\Pi_\rho\) with respect to \(\Pi_1\) is obtained by taking the ratio of
the corresponding densities:
\begin{equation}
\frac{d\Pi_\rho}{d\Pi_1}(\Phi_w)
=
\frac{
j_{\rho,n}(\phi_1,\ldots,\phi_n)
}{
j_{1,n}(\phi_1,\ldots,\phi_n)
}
=
\frac{
\exp(-\rho|\Omega|)\rho^n
}{
\exp(-|\Omega|)
}.
\end{equation}
Thus,
\begin{equation}
\frac{d\Pi_\rho}{d\Pi_1}(\Phi_w)
=
\exp\!\big(-(\rho-1)|\Omega|\big)\rho^{|\Phi_w|}. \label{eq:drho_d1}
\end{equation}

In what follows, we work with this density representation of the PPP prior, implicitly with respect
to the reference measure \(\Pi_1\). Taking the negative logarithm gives
\begin{equation}
-\log \frac{d\Pi_\rho}{d\Pi_1}(\Phi_w)
=
(\rho-1)|\Omega|
-
|\Phi_w|\log\rho.
\end{equation}
Since the first term is independent of \(\Phi_w\), for MAP estimation we may write
\begin{equation}
-\log \frac{d\Pi_\rho}{d\Pi_1}(\Phi_w)
=
\lambda_0|\Phi_w|+\mathrm{const},
\qquad
\lambda_0\triangleq \log(1/\rho)>0.
\end{equation}
Thus, the PPP prior induces a sparsity-promoting penalty proportional to the number of prototypes.

\textbf{Gibbs posterior and MAP objective.}
Using the PPP prior density derived above, we define the Gibbs posterior over
\((\Phi_w,\theta)\) as
\begin{equation}
\pi(\Phi_w,\theta \mid \mu_w)
\propto
\exp\!\Big(
-\beta\,\mathrm{OT}(\mu_w,\nu_{\Phi_w,\theta})
\Big)
\frac{d\Pi_\rho}{d\Pi_1}(\Phi_w)
\pi_0(\theta),
\qquad \beta>0.
\end{equation}
Here the posterior is understood as a density with respect to the reference measure
\(\Pi_1(d\Phi_w)d\theta\). Substituting Eq.~\ref{eq:drho_d1}
\begin{equation}
\frac{d\Pi_\rho}{d\Pi_1}(\Phi_w)
\propto
\rho^{|\Phi_w|}
=
\exp\!\big(-\lambda_0|\Phi_w|\big),
\qquad
\lambda_0=\log(1/\rho)>0,
\end{equation}
gives
\begin{equation}
\pi(\Phi_w,\theta \mid \mu_w)
\propto
\exp\!\Big(
-\beta\,\mathrm{OT}(\mu_w,\nu_{\Phi_w,\theta})
-\lambda_0|\Phi_w|
\Big)
\pi_0(\theta).
\end{equation}

Therefore, the MAP estimator satisfies
\begin{equation}
(\Phi_w^\star,\theta^\star)
\in
\arg\max_{\Phi_w,\theta}
\pi(\Phi_w,\theta \mid \mu_w).
\end{equation}
Equivalently, it minimizes the negative log-posterior density:
\begin{equation}
(\Phi_w^\star,\theta^\star)
\in
\arg\min_{\Phi_w,\theta}
\left[
\beta\,\mathrm{OT}(\mu_w,\nu_{\Phi_w,\theta})
+\lambda_0|\Phi_w|
-\log\pi_0(\theta)
\right],
\end{equation}
up to additive constants independent of \((\Phi_w,\theta)\).

Since \(\beta>0\), dividing by \(\beta\) preserves the minimizer:
\begin{equation}
(\Phi_w^\star,\theta^\star)
\in
\arg\min_{\Phi_w,\theta}
\left[
\mathrm{OT}(\mu_w,\nu_{\Phi_w,\theta})
+\frac{\lambda_0}{\beta}|\Phi_w|
-\frac{1}{\beta}\log\pi_0(\theta)
\right].
\end{equation}
Defining
\begin{equation}
\lambda \triangleq \frac{\lambda_0}{\beta}
=
\frac{\log(1/\rho)}{\beta},
\end{equation}
we obtain
\begin{equation}
(\Phi_w^\star,\theta^\star)
\in
\arg\min_{\Phi_w,\theta}
\left[
\mathrm{OT}(\mu_w,\nu_{\Phi_w,\theta})
+\lambda|\Phi_w|
-\frac{1}{\beta}\log\pi_0(\theta)
\right].
\end{equation}

\textbf{Conclusion.}~Thus, when \(\pi_0(\theta)\) is uniform, MAP inference under the PPP-induced Gibbs posterior coincides with minimizing the cardinality-penalized OT objective
\begin{equation}
\min_{\Phi_w,\theta}
\mathrm{OT}(\mu_w,\nu_{\Phi_w,\theta})
+\lambda|\Phi_w|.
\end{equation}
which coincides with the
penalized OT problem in Eq.~\eqref{eq:penalizedW2}.

\section{Proof of Proposition~\ref{prop:global_to_slide}.} \label{app:prop_proof}
Following formulation in Sec.~\ref{sec:prelim}, we recall the OT functional under square cost:
\begin{equation}
\mathrm{OT}(\alpha,\beta)
=
\inf_{\gamma\in\Gamma(\alpha,\beta)}
\left\{
\int \|u-v\|_2^2\,d\gamma(u,v)
\right\}.
\end{equation}
It is a standard result~\cite{peyre2019computational} that $\mathrm{OT}(\alpha,\beta)$ is jointly convex in $(\alpha,\beta)$ on $\mathcal P_2(\mathbb R^d)$ (for cost $\|u-v\|_2^2$).
Therefore, by Jensen's inequality and the definitions $p_r=\mathbb E_w[\mu_w]$ and $p_\theta=\mathbb E_w[\nu_{w,\theta}]$,
\begin{equation}
\mathrm{OT}(p_r,p_\theta)
=
\mathrm{OT}\!\left(\mathbb E_w[\mu_w],\,\mathbb E_w[\nu_{w,\theta}]\right)
\le
\mathbb E_w\!\left[\mathrm{OT}(\mu_w,\nu_{w,\theta})\right],
\end{equation}
which proves Eq.~\eqref{eq:global_le_slide}.
The penalized inequality Eq.~\eqref{eq:global_penalized_le_slide} follows by applying the same argument to $\nu_{w,\Phi_w,\theta}$ and adding $\lambda\,\mathbb E_w[|\Phi_w|]$ to both sides.

\section{Lemmas and Derivations} \label{app:lemma_proof}

\begin{lemma} (adapted from ~\cite{weng2024probabilistic} )\label{lemma:linear_b}
    For any scalar function $g(\mathbf{r)}$ and a binary vector $\xi = [\xi_1, \xi_2, \ldots, \xi_n]$ such that $\xi_i \in \{0, 1\}$ and $\xi$ has exactly one non-zero component, we have
    \begin{equation}
        g\Bigg(\sum_{i=1}^n \xi_i \cdot \mathbf{r}_i \Bigg) = \sum_{i=1}^n \Bigg( \xi_i \cdot g(\mathbf{r}_i) \Bigg) \label{eq:lemma_a1}
    \end{equation}  
\end{lemma}
with respect to any set $\{ \mathbf{r}_i \}_{i=1}^n$ of valid inputs to $g(\mathbf{r})$.

\textbf{Proof.}~Suppose the only non-zero component appears at position $j$, both sides of Eq.~\ref{eq:lemma_a1} will be $g(\mathbf{r}_j)$.


\begin{lemma} \label{lemma:R1}
    Let $\mathbb{P}(z_m \mid b_m, \Phi)$ defined as in Eq.~\ref{eq:N_new}. Let $R_1(b) \triangleq \sum_{m=1}^M \log \mathbb{P}(z_m \mid b_m, \Phi)$, considering $(z_m, \Phi)$ as constants. We have
    \begin{equation}
        R_1(b) = \sum_{m=1}^M \sum_{k=1}^K b_{mk} \cdot \log \mathcal{N}\left( z_m \mid \phi_k, \text{diag}\left( \vartheta \left( \phi_k; \gamma \right) \right) \right), 
    \end{equation}  
    which is linear in terms of the assignment parameter $b$.
\end{lemma}
\textbf{Proof.}
To derive results of Lemma~\ref{lemma:R1}, note that Eq.~\ref{eq:N_new} implies the following,
\begin{align}
    \log \mathbb{P}(z_m\mid b_m, \Phi) &= \log \mathcal{N}\left( z_m \mid \sum_{k=1}^K b_{mk} \cdot \phi_k, \text{diag}\left( \delta \left( \sum_{k=1}^K b_{mk} \cdot \phi_k; \gamma \right) \right) \right) \\
    & = g\left( \sum_{k=1}^K b_{mk}\cdot \phi_k \right) \label{eq:12}
\end{align}
where we define
\begin{equation}
    g(\mathbf{x}) \triangleq \log \mathcal{N}\bigg( z_m \mid \mathbf{x}, \text{diag}\left( \vartheta( \mathbf{x}; \gamma )\right) \bigg)
\end{equation}
In addition, since $\sum_k b_{mk} =1 $ with $b_{mk} \in \{ 0, 1 \}$, Lemma~\ref{lemma:linear_b} implies that
\begin{equation}
    g\bigg( \sum_{k=1}^K b_{mk} \cdot \phi_k \bigg) = \sum_{k=1}^K \bigg( b_{mk} \cdot g(\phi_k) \bigg) \label{eq:14}
\end{equation} 
We then plug Eq.~\ref{eq:14} into Eq.~\ref{eq:12} to have
\begin{align}
    \log \mathbb{P}(z_m \mid b_m, \Phi) &=  g\bigg( \sum_{k=1}^K b_{mk} \cdot \phi_k \bigg) = \sum_{k=1}^K \bigg( b_{mk} \cdot g(\phi_k)  \bigg) \\
    &= \sum_{k=1}^K b_{mk} \cdot \log \mathcal{N} \bigg( z_m \mid \phi_k, \text{diag}\big( \vartheta(\delta_k; \gamma) \big) \bigg)   \label{eq:16}
\end{align}
Finally, taking summation over $m=1, 2, \ldots , M$ on both sides of Eq.~\ref{eq:16}, we arrive Lemma~\ref{lemma:R1}.

\begin{lemma} \label{lemma:R2}
    Let $\mathbb{P}(b_m \mid \zeta )$ defined as in Eq.~\ref{eq:cate_b}. Let $R_2(b) \triangleq \sum_{m=1}^M \log \mathbb{P}(b_m \mid \zeta)$, considering $(\zeta)$ as constants. We have
    \begin{equation}
        R_2(b) = \sum_{m=1}^M \sum_{k=1}^K b_{mk} \cdot \log \Bigg( \frac{\exp \big( \alpha( \phi_k; \zeta ) \big)}{ \sum_{k'} \exp \big( \alpha(\phi_{k'}; \zeta) \big) } \Bigg), 
    \end{equation}  
    which is linear in terms of the assignment parameter $b$.
\end{lemma}

\textbf{Proof.} Plug Eq.~\ref{eq:cate_b} into the definition of $R_2(b)$, we have
\begin{align}
    R_2(b) &= \sum_{m=1}^M \log \mathbb{P}(b_m \mid \zeta) = \sum_{m=1}^M \log \Bigg( \prod_{k=1}^K \bigg( \frac{\exp\big( \alpha(\phi_k; \zeta) \big)}{\sum_k \exp\big( \alpha( \phi_k; \zeta ) \big)} \bigg)^{b_{mk}} \Bigg) \\
    &= \sum_{m=1}^M \sum_{k=1}^K b_{mk} \cdot \log \Bigg( \frac{\exp\big( \alpha(\phi_k; \zeta) \big)}{\sum_k \exp\big( \alpha( \phi_k; \zeta ) \big)} \Bigg)
\end{align}
which naturally arrives Lemma~\ref{lemma:R2}.

\section{Additional Results} \label{app:more_results}

\subsection{Performance on a Needle-In-A-Hay-Stack Dataset}

\textbf{Setup.}~To further assess robustness beyond tasks where distribution-preserving condensation is naturally effective, we evaluate NICER on Camelyon16~\cite{bejnordi2017camelyon16}, a clinically relevant metastasis detection benchmark where discriminative signals may be sparse. We compare NICER against strong baselines under the same settings and metrics as Tab.~\ref{tab:main-results-dd}, using UNI as the foundation encoder.

\begin{table}[h]
\centering
\caption{Performance of baselines on Condensation tasks on Camelyon 16. \best{The best} and \second{second-best} results are highlighted in \best{bold red}, and \second{blue}.}
\label{tab:camelyon16}
\resizebox{0.8\textwidth}{!}{%
\begin{tabular}{@{}l|c|l|l|l@{}}
\toprule
\multicolumn{1}{c|}{\textbf{Method}} 
& \textbf{Predictor} 
& \multicolumn{1}{c|}{\textbf{Kappa}} 
& \multicolumn{1}{c|}{\textbf{Bal. Acc.}} 
& \multicolumn{1}{c}{\textbf{F1}} \\ 
\midrule

OT 
& \multirow{4}{*}{ABMIL} 
& $-0.0045 \pm 0.0417$ 
& $0.5000 \pm 0.0208$ 
& $0.3193 \pm 0.0915$ \\

H2T 
& 
& $-0.0303 \pm 0.0379$ 
& $0.4837 \pm 0.0204$ 
& $0.4397 \pm 0.0122$ \\

PANTHER 
& 
& \second{$0.1206 \pm 0.1917$} 
& \second{$0.5357 \pm 0.1359$} 
& \second{$0.4759 \pm 0.2117$} \\

NICER (Ours) 
& 
& \best{$0.8471 \pm 0.0512$} 
& \best{$0.9175 \pm 0.0253$} 
& \best{$0.9262 \pm 0.0247$} \\ 
\midrule

OT 
& \multirow{4}{*}{DSMIL} 
& \second{$0.1307 \pm 0.1693$} 
& \second{$0.5735 \pm 0.0936$} 
& $0.4006 \pm 0.1937$ \\

H2T 
& 
& $0.0495 \pm 0.0533$ 
& $0.5575 \pm 0.0308$ 
& $0.4047 \pm 0.0243$ \\

PANTHER 
& 
& $0.0155 \pm 0.0837$ 
& $0.5114 \pm 0.0464$ 
& \second{$0.4352 \pm 0.0504$} \\

NICER (Ours) 
& 
& \best{$0.5453 \pm 0.2233$} 
& \best{$0.7712 \pm 0.0985$} 
& \best{$0.7710 \pm 0.1273$} \\ 
\midrule

OT 
& \multirow{4}{*}{ILRA} 
& \second{$0.1856 \pm 0.1702$} 
& \second{$0.6017 \pm 0.0927$} 
& \second{$0.5674 \pm 0.0743$} \\

H2T 
& 
& $-0.0259 \pm 0.1092$ 
& $0.4853 \pm 0.0530$ 
& $0.4528 \pm 0.0283$ \\

PANTHER 
& 
& $-0.0171 \pm 0.0296$ 
& $0.4910 \pm 0.0156$ 
& $0.3802 \pm 0.1199$ \\

NICER (Ours) 
& 
& \best{$0.7254 \pm 0.1014$} 
& \best{$0.8762 \pm 0.0530$} 
& \best{$0.8627 \pm 0.0504$} \\ 
\bottomrule
\end{tabular}%
}
\end{table}

\textbf{Discussion.}~As shown in Tab.~\ref{tab:camelyon16}, NICER consistently achieves the best performance across all metrics and predictors. In fact, NICER outperforms the baselines by up to approximately 43\% and 39\% in accuracy using ABMIL and ILRA, respectively. This demonstrates NICER's ability to preserve diagnostically important features even when the relevant patches constitute a small minority of the whole slide. In contrast, other baselines such as PANTHER lack capacity-adaptive properties and consequently fail to retain such features, resulting in more substantial performance degradation. 

\subsection{Extended Results with a Different Pretrained Encoder}

\textbf{Setup.}~We further evaluated our algorithm against the existing SOTA baselines (using the same settings and metrics as Table 1) on the PANDA dataset using CONCH~\citep{lu2024conch} as the pre-trained encoder in order to investigate their robustness across different pretrained architectures. It is a SOTA vision-language foundation model whose multimodal features provide a stronger and more diverse evaluation settings in computational pathology.

\textbf{Discussion.}~As shown in Tab.~\ref{tab:results-conch}, NICER consistently achieves the best performance across all metrics and predictors. In fact, NICER outperforms the most competitive baselines by around 23\% and 13\% in accuracy using ABMIL and DSMIL, respectively. This demonstrates NICER's ability to preserve diagnostically important features constructed by different encoders with different pretraining paradigms, highlighting the generalization of our method.

\begin{table}[]
\centering
\caption{Performance of baselines on Condensation tasks on PANDA dataset using CONCH as a pretrained encoder. \best{The best} and \second{second-best} results are highlighted in \best{bold red}, and \second{blue}.}
\label{tab:results-conch}
\resizebox{0.85\textwidth}{!}{%
\begin{tabular}{@{}l|c|l|l|l@{}}
\toprule
\multicolumn{1}{c|}{\textbf{Method}} & \textbf{Predictor}     & \multicolumn{1}{c|}{\textbf{Kappa}} & \multicolumn{1}{c|}{\textbf{Accuracy}} & \multicolumn{1}{c}{\textbf{F1}} \\ \midrule
OT                                   & \multirow{4}{*}{ABMIL} & 0.0264 ± 0.0849                     & 0.2036 ± 0.0158                        & 0.1213 ± 0.0318                 \\
H2T                                  &                        & \second{0.5318 ± 0.1789}                     & \second{0.4034 ± 0.0736}                        & \second{0.3439 ± 0.1228}                 \\
PANTHER                              &                        & 0.0467 ± 0.1002                     & 0.2373 ± 0.0807                        & 0.1341 ± 0.0659                 \\
\textbf{NICER (Ours)}                &                        & \best{0.8392 ± 0.0514}            & \best{0.6398 ± 0.0365}               & \best{0.6332 ± 0.0327}        \\ \midrule
OT                                   & \multirow{4}{*}{DSMIL} & {0.6253 ± 0.0944}                     & \second{0.4430 ± 0.0636}                        & \second{0.4436 ± 0.0843}                 \\
H2T                                  &                        & \second{0.6548 ± 0.0411 }                    & 0.4183 ± 0.0083                        & 0.3494 ± 0.0205                 \\
PANTHER                              &                        & 0.0416 ± 0.1444                     & 0.1782 ± 0.0587                        & 0.0834 ± 0.0538                 \\
\textbf{NICER (Ours)}                &                        & \best{0.8283 ± 0.0230}            & \best{0.5712 ± 0.0237}               & \best{0.5675 ± 0.0311}        \\ \bottomrule
\end{tabular}
}
\end{table}

\subsection{Prototype Analysis}

\textbf{Concept Prototype Diversity and Convergence.}
We extend the findings of Section~\ref{sec:main_results} by examining additional WSI instances from TCGA-NSCLC. Figure~\ref{fig:tsne_prototypes_track_full} illustrates the trajectories of concept prototypes over 10 optimization iterations of NICER for three representative samples: \textit{TCGA-93-A4JN-01Z-00-DX1.ED4C9365-6CCF-4AEE-B4C9-3CC5EC57286C}, \textit{TCGA-50-6594-01Z-00-DX1.43b2005a-4245-4025-ad85-4a957f308a5c}, and \textit{TCGA-49-4514-01Z-00-DX2.f1565a36-257d-432e-a84d-47c1d7a0185f}. The visualizations reveal that different WSIs exhibit distinct condensation dynamics, reflecting variation in morphological complexity and feature distribution. For example, the first sample (ending in “86C”) shows prototypes that initially cluster tightly, suggesting greater homogeneity, whereas the other two slides begin with more diffuse clusters, indicating higher heterogeneity.

\begin{figure}
    \centering
    \includegraphics[width=0.95\linewidth]{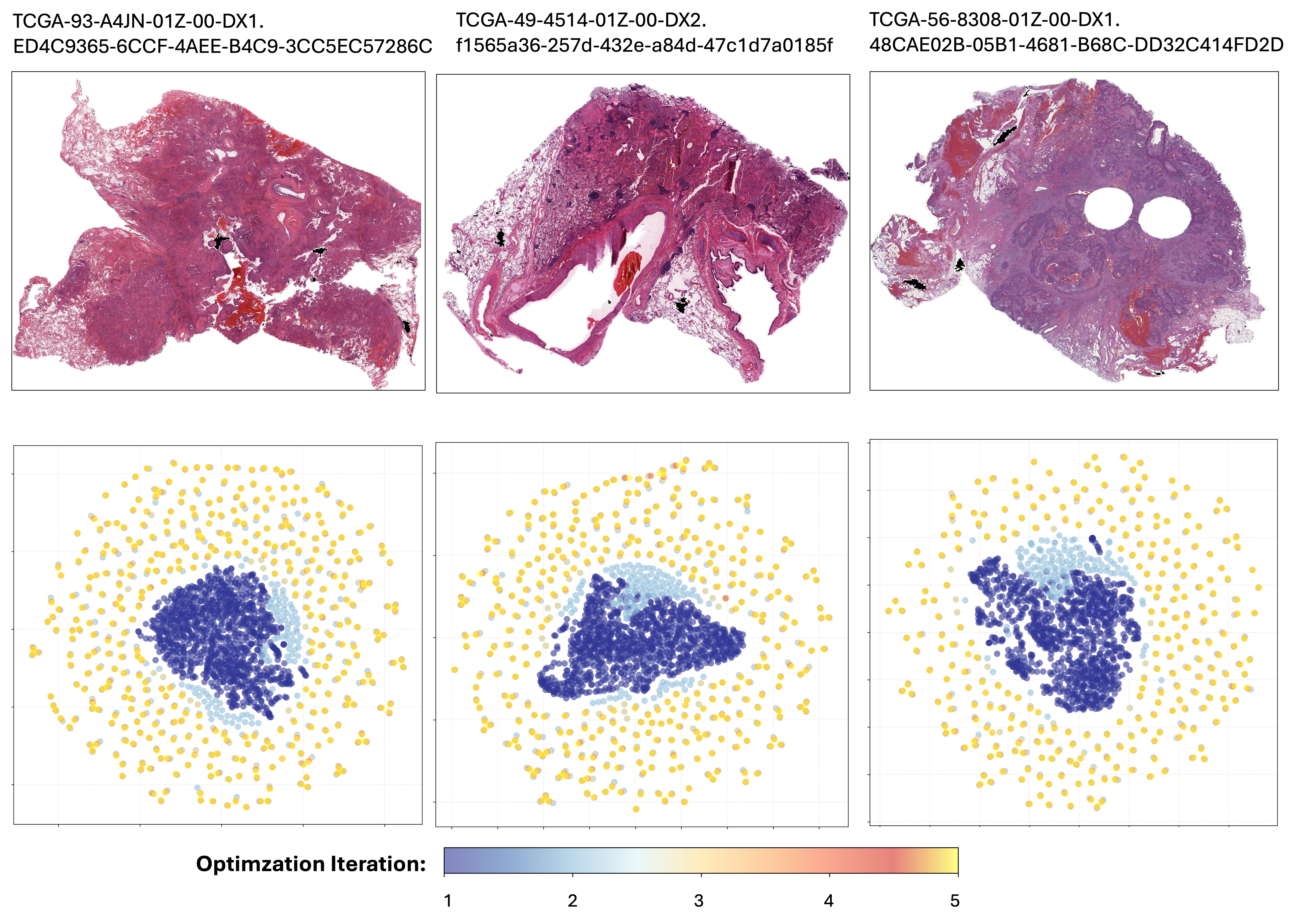}
    \caption{2D t-SNE visualizations of NICER's concept prototypes learned over five iterations using three representative TCGA slides. These plots highlight that NICER can capture different diverse information across different WSIs.}
    \label{fig:tsne_prototypes_track_full}
\end{figure}

Despite these sample-specific differences, a consistent pattern emerges across all trajectories: prototypes gradually diverge from their initialization with increasing variance, reflecting how condensation enforces specialization and reduces redundancy while preserving diversity. This behavior highlights NICER’s ability to uncover distinct and non-overlapping concept structures within each slide. Complementary evidence is shown in Figure~\ref{fig:number_of_proto_full}, where the number of prototypes stabilizes after only a few iterations. This rapid convergence indicates that redundant concepts are pruned early, leaving a compact and stable set that continues to refine qualitatively rather than quantitatively. Together, these results emphasize NICER’s efficiency in learning diverse, non-redundant representations of WSIs with minimal optimization steps.

\begin{figure}
    \centering
    \includegraphics[width=0.9\linewidth]{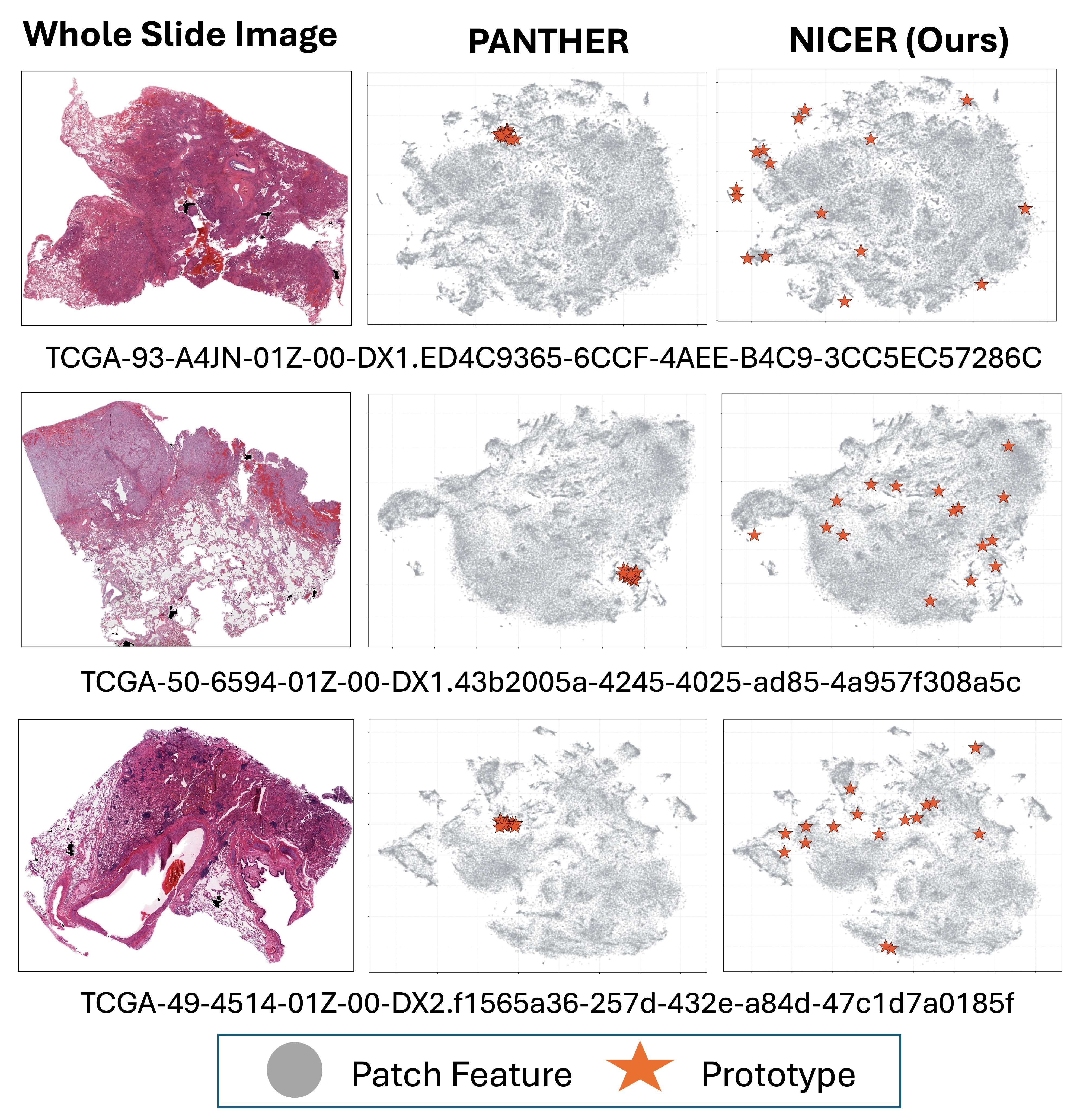}
    \caption{2D t-SNE visualization of original WSI feature bags and learned concept prototypes, comparing our method (NICER) with a representative prior approach (PANTHER).}
    \label{fig:prototype_clusters}
\end{figure}

\textbf{Diversity Preservation.}
To evaluate NICER’s ability to preserve diversity across whole-slide images (WSIs), we compare 2D t-SNE visualizations of concept prototypes generated by NICER and PANTHER, alongside the original feature bags, on three TCGA-NSCLC samples: \textit{TCGA-93-A4JN-01Z-00-DX1.ED4C9365-6CCF-4AEE-B4C9-3CC5EC57286C}, \textit{TCGA-50-6594-01Z-00-DX1.43b2005a-4245-4025-ad85-4a957f308a5c}, and \textit{TCGA-49-4514-01Z-00-DX2.f1565a36-257d-432e-a84d-47c1d7a0185f}. For fair visualization, NICER’s prototypes are clustered with K-Means to 16, matching the fixed prototype count used in PANTHER~\citep{song2024cvpr}. As shown in Figure~\ref{fig:prototype_clusters}, PANTHER fails to capture the inherent diversity of WSIs, collapsing heterogeneous regions into a limited number of clusters and discarding critical information required for downstream tasks, an effect consistent with its suboptimal performances in Table~\ref{tab:main-results-dd} and Table~\ref{tab:morphology_prototype}. This limitation arises because PANTHER enforces a rigid and heuristically small prototype budget, prioritizing efficiency at the expense of representational fidelity. In contrast, NICER employs an alternating optimization strategy that leverages pattern-based diversity preservation and condensation-driven efficiency, ensuring prototypes remain well-separated and encode distinct, slide-specific information. This balance enables NICER to adapt to the variability of individual WSIs while maintaining compact yet expressive concept sets.

To validate the consistency of our observations in Fig.~\ref{fig:prototype_clusters}, we further project the original feature bag and the concept prototypes learned by NICER and PANTHER for a WSI from the TCGA-BRCA cohort (slide ID: \textit{TCGA-49-4514-01Z-00-DX2.f1565a36-257d-432e-a84d-47c1d7a0185f}) using three 2D projection methods: t-SNE, UMAP, and PCA. As shown in Fig.~\ref{fig:different_2D}, NICER preserves a diverse set of pathological patterns after the condensation process. In contrast, PANTHER yields highly similar prototype representations with limited discriminative structure, indicating substantial information loss. This observation aligns with ours from Fig.~\ref{fig:prototype_clusters} and downstream results shown in Sec.~\ref{sec:experiments}.

{
\textbf{Concept Prototype Capacity Analysis}. To evaluate NICER's ability to adapt its capacity to complexity of data, we report the average number of concept prototypes across different datasets. As can be seen from Table\ref{tab:num_of_proto}, for larger and more complex WSIs (e.g., NSCLC, LUAD, BRCA), NICER naturally allocates more capacity to preserve the underlying conceptual diversity in the slides. This aligns with our observation in Section~\ref{sec:experiments}.
}

\begin{table}[t]
\caption{{Average number of final concept prototypes across different datasets}}
\centering
\label{tab:num_of_proto}
\begin{tabular}{@{}llll@{}}
\toprule
PANDA & NSCLC  & LUAD (Survival) & BRCA (Survival) \\ \midrule
62.80 & 154.67 & 157.40          & 158.60          \\ \bottomrule
\end{tabular}
\end{table}

\begin{figure}
   \begin{minipage}[]{0.45\textwidth}
        \centering
        \captionof{table}{Impact of number of initial patterns $M$ on overall performance across predictors}
        \label{tab:abl_initialP_full}
        \resizebox{0.90\linewidth}{!}{%
        
            \begin{tabular}{@{}c|c|ccc@{}}
            \toprule
            \textbf{Predictor} & \textbf{\begin{tabular}[c]{@{}c@{}}$M$ \end{tabular}} & \textbf{Kappa} & \textbf{Accuracy} & \textbf{F1} \\ \midrule
            \multirow{4}{*}{\textbf{ABMIL}} & 50   & 85.48 & 67.56 & 67.64 \\
                                            & 200  & 90.36 & 73.15 & 73.04 \\
                                            & 500  & 91.85 & 74.50 & 74.46 \\
                                            & 1000 & 91.48 & 74.94 & 75.00 \\ \midrule
            \multirow{4}{*}{\textbf{DSMIL}} & 50   & 82.75 & 60.18 & 61.03 \\
                                            & 200  & 86.74 & 66.89 & 64.56 \\
                                            & 500  & 88.59 & 67.34 & 68.03 \\
                                            & 1000 & 88.73 & 68.46 & 67.52 \\ \midrule
            \multirow{4}{*}{\textbf{ILRA}}  & 50   & 84.41 & 65.55 & 65.13 \\
                                            & 200  & 89.67 & 66.62 & 67.34 \\
                                            & 500  & 92.25 & 73.83 & 74.55 \\
                                            & 1000 & 92.25 & 73.83 & 74.55 \\ \bottomrule
            \end{tabular}%
        }
    \end{minipage}
    \hfill
    \begin{minipage}[]{0.50\textwidth}
        \centering
            \includegraphics[width=\linewidth]{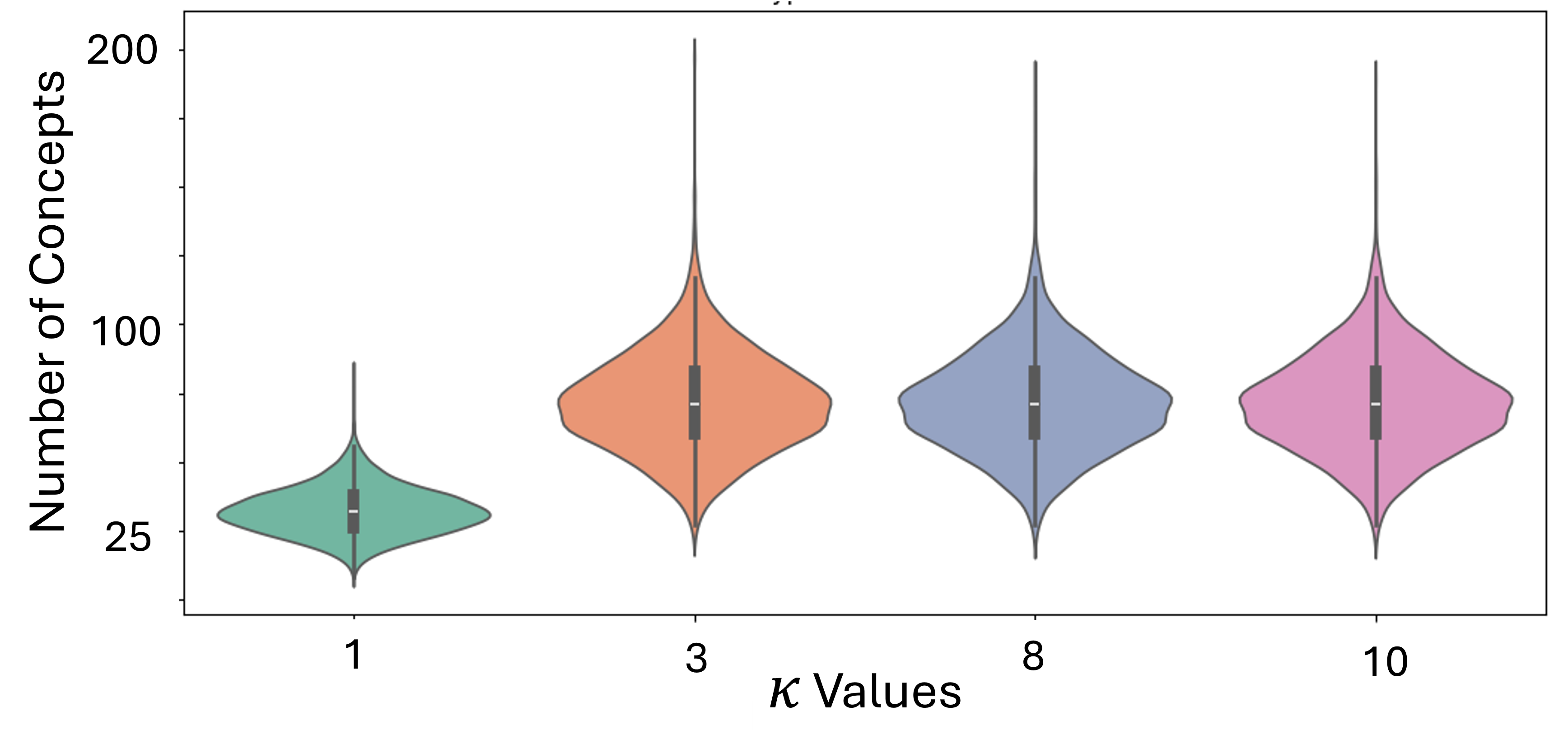}
            \captionof{figure}{Impact of top-$\kappa$ in number of final concept prototypes, conducted on PANDA. The plots show that number of concepts increases when we increase $\kappa$, allowing more preservation capacity. After a certain point, the condensation process saturates this count to a stable value range.}
    \label{fig:nproto_by_k}
    \end{minipage}
\end{figure}

\subsection{Sensitivity Analysis}

\begin{figure}
    \centering
    \includegraphics[width=0.9\linewidth]{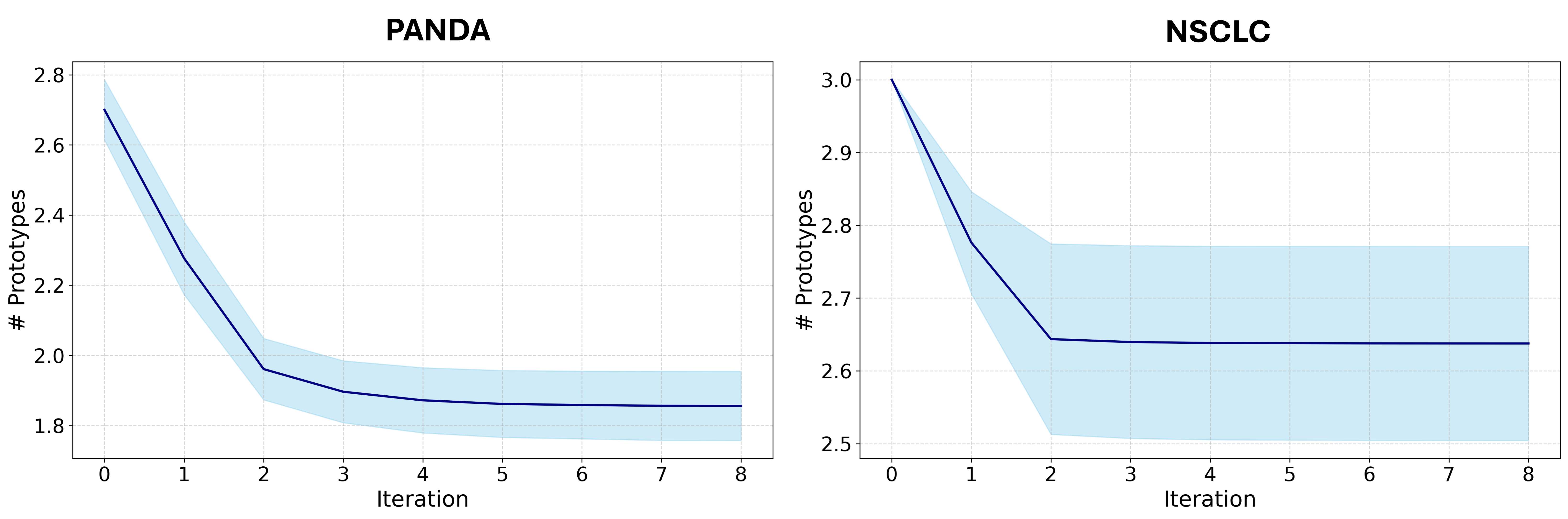}
    \caption{Number of prototypes tracked over the first 8 iterations on PANDA and NSCLC dataset. The values are reported at log10 scale}
    \label{fig:number_of_proto_full}
\end{figure}

\begin{table}[]{}
    \centering
    \captionof{table}{Performance-Efficiency trade-offs comparison of NICER and the SOTA method, on PANDA using ABMIL under the same seed.}
    
    \resizebox{0.5\linewidth}{!}{%
       \begin{tabular}{@{}c|lcl@{}}
        \toprule
        \textbf{Method}                                                   & \multicolumn{1}{c}{\textbf{Config}} & \textbf{\% Input} & \multicolumn{1}{c}{\textbf{F1}} \\ \midrule
        \begin{tabular}[c]{@{}c@{}}Whole Bag \\ (Upperbound)\end{tabular} & \multicolumn{1}{c}{-}               & 100\%                      & 74.42                           \\ \midrule
        \multirow{3}{*}{PANTHER}     & K=16   & 2.58\%  & 41.11 \\
                                     & K=64   & 10.39\% & 55.93 \\
                                     & K=128  & 20.78\% & 69.45 \\ \midrule
        \multirow{4}{*}{NICER (Ours)} & M=50   & 5.79\%  & 67.64 \\
                                     & M=200  & 11.32\% & 73.04 \\
                                     & M=500  & 11.56\% & 74.46 \\
                                     & M=1000 & 11.31\% & 75.00 \\ \bottomrule
        \end{tabular}%
    } \label{tab:performance_efficiency}
\end{table}

\textbf{Extensive Results on Impacts of Number of Initial Prototypes ($M$).}
Table~\ref{tab:abl_initialP_full} complements the analysis in Figure~\ref{fig:initial_P} by evaluating how the initial number of patterns influences NICER’s performance across different predictor architectures. The results demonstrate that increasing the pattern set size consistently improves performance, as a larger pool enhances the model’s capacity to distill informative representations. However, these gains diminish once the pattern set grows beyond a certain threshold, with performance gradually reaching a plateau. This saturation reflects the role of the condensation process, which systematically eliminates redundancy and stabilizes the number of effective representations required to characterize a WSI, regardless of the starting pattern count, aligning our insights discussed from Figure~\ref{fig:initial_P}.

\textbf{Extensive Results on Impacts of Top-$\kappa$ during Pattern Exploration.}
We further conduct a sensitivity analysis on the number of pattern assigned to each patch feature, controlled by the top-$\kappa$ selection in the Condensation Ability tasks. The results, reported in Tab.~\ref{tab:abl_top_k_full}, reveal a clear trend: model performance improves steadily across all evaluation metrics as $\kappa$ increases, indicating that incorporating multiple patterns per patch allows richer information to be preserved. However, beyond a moderate value (e.g., $\kappa=3$), this benefit begins to diminish, with performance gains plateauing or even slightly decreasing. The degradation at larger $\kappa$ arises because information from each patch becomes overly dispersed across many patterns, weakening the sharpness of the distilled representation. These observations corroborate our earlier findings in Fig.~\ref{fig:initial_P} and Tab.~\ref{tab:abl_top_k_perform}, and reinforce the design principle outlined in Sec.~\ref{sec:distillation} that pattern assignments must balance informativeness with compactness.

\textbf{Performance-Efficiency Trade-offs Comparison.}
To assess the effectiveness of NICER in the condensation problem, we compare its performance-efficiency trade-off against PANTHER on the cancer subtyping task using the PANDA dataset across varying condensation levels. The condensation level is controlled by the number of prototypes $K$ in PANTHER and the number of initial patterns $M$ in NICER, where $M$ serves as an upper bound on the number of final concepts. As reported in Tab.~\ref{tab:performance_efficiency}, both methods exhibit an upward trend in F1 as the \% Input increases. Notably, at comparable condensation levels (e.g., PANTHER with 10.39\% Input vs.\ NICER with 11.32\% Input), NICER surpasses PANTHER by a substantial margin of nearly 19\%, underscoring its superior ability to preserve relevant information during condensation.

\begin{table}[t]
\centering
\caption{Ablation study of top-$\kappa$ on downstream tasks performance across different predictor architectures. The experiments are conducted on PANDA dataset.}
\label{tab:abl_top_k_full}
\resizebox{0.85\linewidth}{!}{%
\begin{tabular}{@{}c|ccc|ccc|ccc@{}}
\toprule
\textbf{Predictor} & \multicolumn{3}{c|}{\textbf{ABMIL}} & \multicolumn{3}{c|}{\textbf{DSMIL}} & \multicolumn{3}{c}{\textbf{ILRA}} \\ \midrule
Top-$\kappa$ & Kappa & Acc.  & F1    & Kappa & Acc.  & F1    & Kappa & Acc.  & F1    \\ \midrule
1     & 90.18 & 70.02 & 69.87 & 85.42 & 64.65 & 65.95 & 89.17 & 64.88 & 64.24 \\
3     & 91.48 & 74.94 & 75.00 & 88.73 & 68.46 & 67.52 & 92.25 & 73.83 & 74.55 \\
5     & 90.47 & 72.93 & 72.72 & 88.80 & 68.46 & 66.66 & 88.66 & 71.36 & 71.47 \\
10    & 89.52 & 70.47 & 70.49 & 85.42 & 66.89 & 66.94 & 88.88 & 68.90 & 68.92 \\ \bottomrule
\end{tabular}%
}
\end{table}

\textbf{Extensive Results on Effectiveness of Condensation Stage.}
Fig.~\ref{fig:nproto_by_k} shows the distribution of the number of final concepts produced by NICER across different $\kappa$ values. As expected, the average number of prototypes increases with larger $\kappa$, since greater capacity enables the model to capture more diverse patterns from complex WSIs. Beyond a certain point, however, this growth plateaus, indicating that the condensation process has effectively merged overlapping patterns and compressed them into a compact, stable concept set. This stabilization highlights NICER’s ability to balance capacity with redundancy removal, yielding a consistent number of meaningful prototypes.




\subsection{Performance under Controlled Compute Settings}

We further investigate whether the performance gains of NICER can be explained by differences in compression budget. Specifically, we compare NICER against fixed-prototype baselines under matched prototype budgets. For each dataset, we rerun the selected baselines using the average and maximum numbers of prototypes discovered by NICER. These settings provide the baselines with an equal or larger prototype budget than NICER, enabling a controlled evaluation of whether NICER’s gains persist after accounting for prototype count. Because NICER determines the number of concepts adaptively, these matched budgets are not available a priori and are used only for this controlled analysis.

\paragraph{Mean-prototype setting.} We compute the mean number of prototypes discovered by NICER across the dataset and use it as a fixed budget for the baselines, thereby matching their global dataset-level budget to NICER’s. The resuls are presented in Tab.~\ref{tab:app:mean_prototype}.

\begin{table}[]
\caption{Performance comparison of NICER and selected baselines in Mean-Prototype settings}
\label{tab:app:mean_prototype}
\centering
\begin{tabular}{@{}l|c|ccc|c@{}}
\toprule
\multicolumn{1}{c|}{\multirow{2}{*}{\textbf{Method}}} & \multirow{2}{*}{\textbf{Head}} & \multicolumn{3}{l|}{\textbf{NSCLC - Match}}   & \multicolumn{1}{l}{\textbf{LUAD-Match}} \\ \cmidrule(l){3-6} 
\multicolumn{1}{c|}{}                                 &                                & Kappa         & Bal. Acc.     & F1            & C-Index                                 \\ \midrule
PANTHER                                               & \multirow{3}{*}{ABMIL}         & 9.67 ± 3.23   & 54.84 ± 1.59  & 52.51 ± 4.27  & 52.79 ± 10.26                           \\
OT                                                    &                                & 82.04 ± 2.96  & 91.03 ± 1.54  & 91.02 ± 1.47  & 62.60 ± 2.22                            \\
NICER                                                 &                                & 89.73 ± 1.81  & 94.84 ± 0.92  & 94.87 ± 0.91  & 64.67 ± 2.11                            \\ \midrule
PANTHER                                               & \multirow{3}{*}{DSMIL}         & 1.21 ± 2.01   & 48.70 ± 2.66  & 42.41 ± 8.36  & 40.03 ± 3.11                            \\
OT                                                    &                                & 78.74 ± 12.13 & 89.24 ± 6.15  & 89.20 ± 6.33  & 56.54 ± 1.61                            \\
NICER                                                 &                                & 84.61 ± 1.56  & 92.30 ± 0.74  & 91.34 ± 0.79  & 62.12 ± 1.41                            \\ \midrule
PANTHER                                               & \multirow{3}{*}{ILRA}          & 4.51 ± 7.81   & 52.58 ± 3.66  & 43.95 ± 11.50 & 53.78 ± 6.59                            \\
OT                                                    &                                & 63.30 ± 20.57 & 81.78 ± 10.08 & 54.11 ± 46.25 & 61.45 ± 6.77                            \\
NICER                                                 &                                & 88.47 ± 1.56  & 94.28 ± 0.76  & 94.23 ± 0.79  & 68.06 ± 4.69                            \\ \bottomrule
\end{tabular}
\end{table}

\paragraph{Maximum-prototype setting.} To stress test this further, we conduct additional ablation studies where we set each baseline’s budget to the maximum number of prototypes discovered by NICER, eliminating NICER’s slide-level budget advantage (the baseline’s prototype budget is now larger than NICER’s). The results are reported in Tab.~\ref{tab:app:max_prototype}.

\begin{table}[]
\caption{Performance comparison of NICER and selected baselines in Max-Prototype settings}
\label{tab:app:max_prototype}
\centering
\begin{tabular}{@{}l|c|lll@{}}
\toprule
\multicolumn{1}{c|}{\multirow{2}{*}{\textbf{Method}}} & \multirow{2}{*}{\textbf{Head}} & \multicolumn{3}{l}{\textbf{NSCLC - Max}}     \\ \cmidrule(l){3-5} 
\multicolumn{1}{c|}{}                                 &                                & Kappa         & Bal. Acc.    & F1            \\ \midrule
PANTHER                                               & \multirow{3}{*}{ABMIL}         & 2.35 ± 5.56   & 51.19 ± 2.80 & 41.52 ± 7.33  \\
OT                                                    &                                & 73.84 ± 15.46 & 87.01 ± 7.60 & 86.51 ± 8.41  \\
NICER                                                 &                                & 89.73 ± 1.81  & 94.84 ± 0.92 & 94.87 ± 0.91  \\ \midrule
PANTHER                                               & \multirow{3}{*}{DSMIL}         & 6.97 ± 0.96   & 53.46 ± 0.47 & 44.10 ± 3.38  \\
OT                                                    &                                & 82.06 ± 5.54  & 91.06 ± 2.75 & 91.02 ± 2.78  \\
NICER                                                 &                                & 84.61 ± 1.56  & 92.30 ± 0.74 & 91.34 ± 0.79  \\ \midrule
PANTHER                                               & \multirow{3}{*}{ILRA}          & 4.41 ± 6.03   & 52.21 ± 3.00 & 43.00 ± 10.91 \\
OT                                                    &                                & 72.56 ± 19.20 & 86.32 ± 9.44 & 85.78 ± 10.46 \\
NICER                                                 &                                & 88.47 ± 1.56  & 94.28 ± 0.76 & 94.23 ± 0.79  \\ \bottomrule
\end{tabular}
\end{table}

\paragraph{Observation.} In both settings, NICER consistently outperforms both PANTHER and OT. This confirms that the performance gain indeed comes from NICER design rather than from leveraging more prototype budget.
Notably, the results also show that the performance of PANTHER and OT do not increase with more prototype budget (see Tab.~\ref{tab:app:mean_prototype}). This observation is consistent with prior work~\cite{song2024cvpr}, which reports that the performance of fixed-budget condensation methods does not necessarily improve with larger storage/prototype budgets and may even degrade. This is because a globally fixed budget cannot account for variations in slide complexity: it may allocate excessive capacity to simple slides while providing insufficient capacity for complex ones. This mismatch can introduce redundant or biased features in the former case and information loss in the latter, limiting transferability across tasks and prediction heads. These results highlight the need for a condensation framework with slide-adaptive capacity, such as NICER.

\section{Pathological Interpretation Details} \label{app:biology}

We investigate whether the condensed representation set learned To identify interpretable patterns, we apply K-means clustering to the condensed features with $K=16$, following \cite{song2024cvpr}, and expect the condensed representations to correctly map the original patches to their corresponding clusters. As described in Sec.~\ref{sec:method}, the correspondence between original patches and global concepts remains traceable after condensation. Consequently, the K-means cluster assignments on the condensed features can be propagated back to group the original patches. Intuitively, if the condensed concept set retains essential pathological information, clustering in this condensed space should yield accurate pathology pattern assignments, without the need to directly cluster the raw patch features. 

Specifically, we conduct experiments at two different scales.
\textbf{(1) Local pattern annotation}, where we annotate representative pathological patterns associated with each cluster to assess the intra-pattern interpretability of the learned concepts; and
\textbf{(2) Global pattern visualization}, where we examine inter-pattern discrepancies by visualizing cluster assignments overlaid on the original WSIs.
In the global analysis, we compare our method with PANTHER, a state-of-the-art approach, to demonstrate NICER’s superior ability to preserve pathological information. All experiments are conducted and evaluated by a board-certified pathologist.

\begin{figure}
    \centering
    \includegraphics[width=\linewidth]{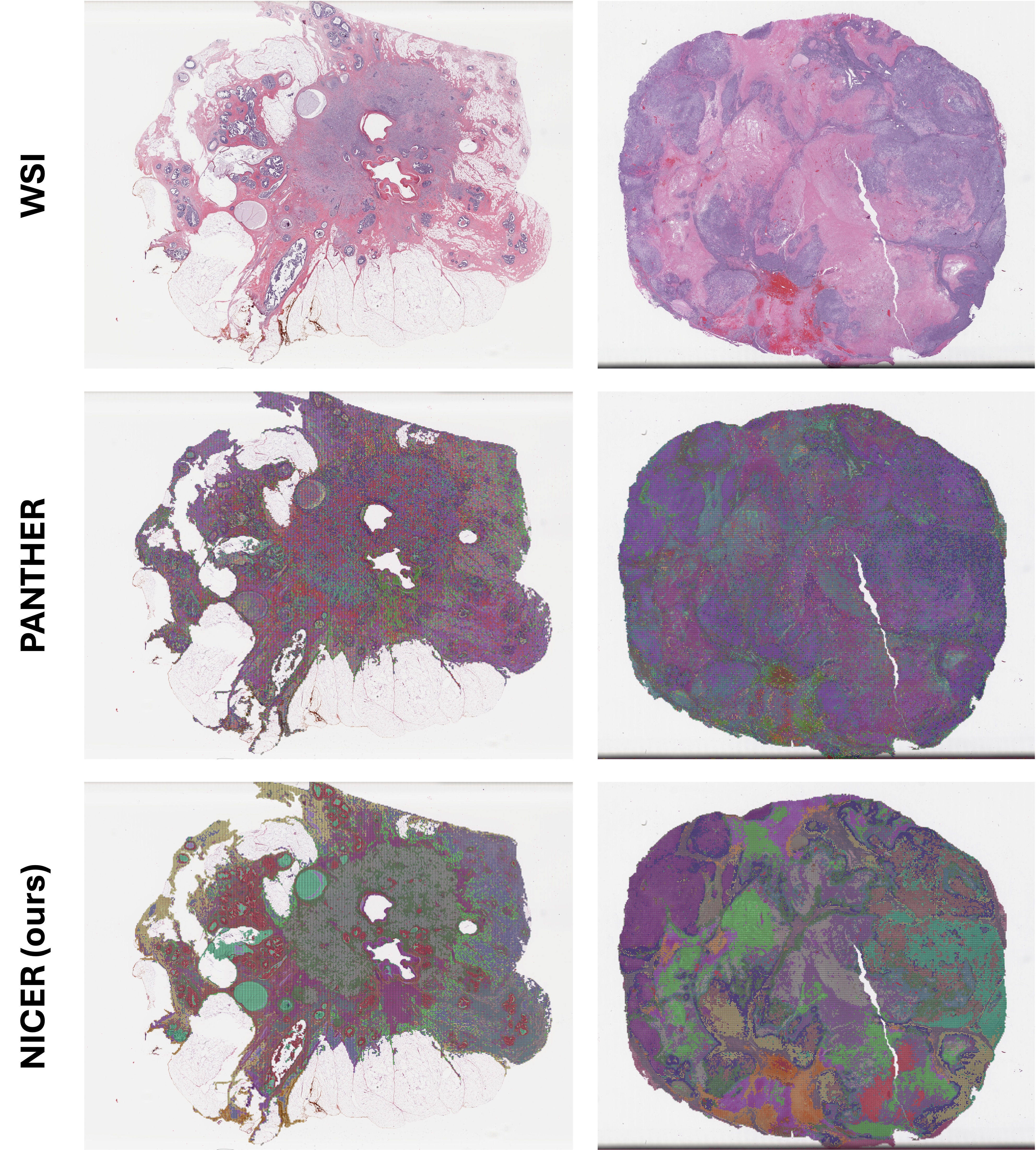}
    \caption{Visualization of prototype assignments on whole-slide images (WSIs) produced by PANTHER and NICER (ours).
    From left to right, the WSIs correspond to \textit{TCGA-E2-A570-01Z-00-DX1.64E5F2CA-9396-40DC-A300-04A54657BDE6} and \textit{TCGA-LL-A8F5-01Z-00-DX1.7579425E-1425-4160-A9DD-3D50F4C5428D} from the TCGA-BRCA cohort. NICER yields more coherent prototype assignments and more accurately delineates tumor regions in both cases, highlighting its strong unsupervised learning capability. \textbf{Best to view in color}}
    \label{fig:segment_visualization}
\end{figure}

\begin{figure}
    \centering
    \includegraphics[width=\linewidth]{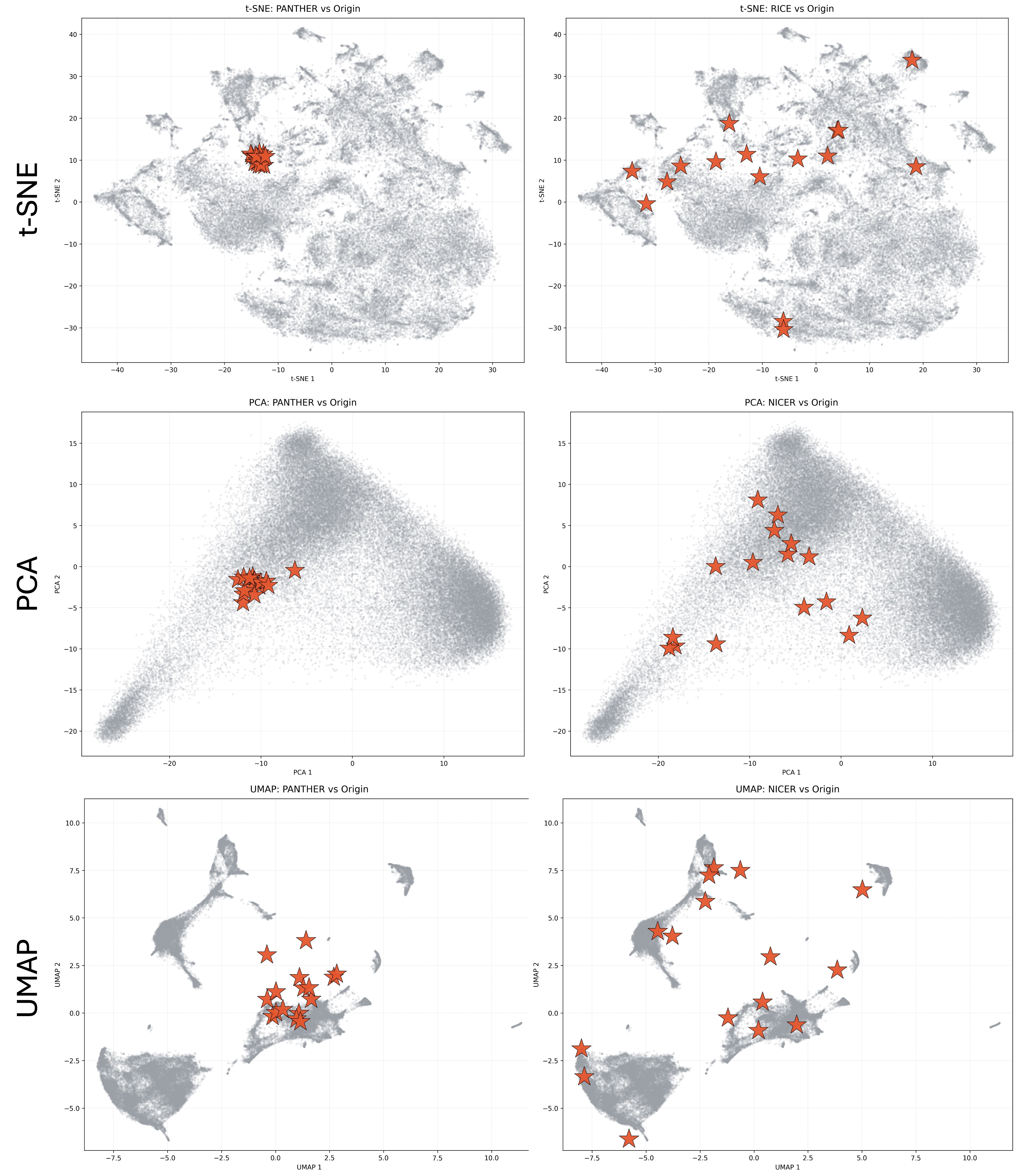}
    \caption{Different 2D projections of original WSI feature bags and learned concept prototypes, showing that our observation on diversity preservation is consistent across multiple projection space. The experiment is conducted on a WSI in TCGA-BRCA with the slide id \textit{TCGA-49-4514-01Z-00-DX2.f1565a36-257d-432e-a84d-47c1d7a0185f}}
    \label{fig:different_2D}
\end{figure}

\subsection{Local Pattern Annotation}
As shown in Fig.~\ref{fig:biology}, the local pattern annotations demonstrate that NICER-derived clusters correspond to well-defined and clinically recognizable histopathological entities, including tumor cell populations with pleomorphic nuclei, mitotic figures, apoptotic cells, necrotic areas, inflammatory infiltrates, vascular structures, and background components such as red blood cells and debris. Importantly, each cluster exhibits high intra-cluster morphological consistency, indicating that the condensed representations preserve fine-grained cellular and tissue-level features. Moreover, NICER is able to separate closely related but diagnostically distinct patterns (e.g., viable tumor versus tumor-adjacent necrosis, dispersed tumor cells versus dense tumor nests), which are critical for pathological assessment. These results suggest that the learned concepts are not only compact but also biologically meaningful, supporting the potential of NICER for clinically interpretable morphological representation learning.

\subsection{Global Pattern Visualization}
Fig.~\ref{fig:segment_visualization} shows the visualization of all NICER's patterns clustered by K-Means and PANTHER~\cite{song2024cvpr}. We observe that compared to PANTHER, NICER produces pattern maps that align more closely with the underlying histological structures observed in the original WSIs. In particular, regions corresponding to tumor tissue show stronger spatial agreement with the WSI morphology, while necrotic and non-tumor areas are more coherently grouped. In contrast, PANTHER exhibits more fragmented and spatially noisy assignments, making direct pathological interpretation more challenging. These observations suggest that NICER better preserves biologically meaningful spatial organization during global pattern mapping.

\section{Limitations} \label{app:limitations}

While NICER effectively handles heterogeneous WSI complexity under limited computational resources, several limitations remain. 

First, NICER adopts a bag-of-features assumption, treating patches independently without explicitly modeling spatial or multi-scale relationships. While this aligns with assumptions of standard MIL settings and is suitable for our tasks, fine-grained tissue architecture and spatial priors are not directly captured, which may be important for clinically sensitive tasks. Incorporating spatial structure into WSI condensation is an important direction for future work.

Second, the PPP prior assumes independent concept occurrence. While this is reasonable when concepts represent distinguishable prototypical patterns, it may limit the ability to capture highly complex or correlated morphological structures in extreme cases. Extending NICER with more flexible priors could further improve its expressiveness.

Finally, NICER is formulated under a fixed pretrained representation lens. Although we examine the effect of encoder choice using two foundation models, a more systematic evaluation across diverse pretrained encoders would further clarify its robustness, generality, and applicability in real-world medical settings.

\newpage
\section*{NeurIPS Paper Checklist}

\begin{enumerate}

\item {\bf Claims}
    \item[] Question: Do the main claims made in the abstract and introduction accurately reflect the paper's contributions and scope?
    \item[] Answer: \answerYes{} 
    \item[] Justification: We present practical motivation in Sec.~\ref{sec:intro} for our algorithm design.
    \item[] Guidelines:
    \begin{itemize}
        \item The answer \answerNA{} means that the abstract and introduction do not include the claims made in the paper.
        \item The abstract and/or introduction should clearly state the claims made, including the contributions made in the paper and important assumptions and limitations. A \answerNo{} or \answerNA{} answer to this question will not be perceived well by the reviewers. 
        \item The claims made should match theoretical and experimental results, and reflect how much the results can be expected to generalize to other settings. 
        \item It is fine to include aspirational goals as motivation as long as it is clear that these goals are not attained by the paper. 
    \end{itemize}

\item {\bf Limitations}
    \item[] Question: Does the paper discuss the limitations of the work performed by the authors?
    \item[] Answer: \answerYes{} 
    \item[] Justification: Please refer to App.~\ref{app:limitations}
    \item[] Guidelines:
    \begin{itemize}
        \item The answer \answerNA{} means that the paper has no limitation while the answer \answerNo{} means that the paper has limitations, but those are not discussed in the paper. 
        \item The authors are encouraged to create a separate ``Limitations'' section in their paper.
        \item The paper should point out any strong assumptions and how robust the results are to violations of these assumptions (e.g., independence assumptions, noiseless settings, model well-specification, asymptotic approximations only holding locally). The authors should reflect on how these assumptions might be violated in practice and what the implications would be.
        \item The authors should reflect on the scope of the claims made, e.g., if the approach was only tested on a few datasets or with a few runs. In general, empirical results often depend on implicit assumptions, which should be articulated.
        \item The authors should reflect on the factors that influence the performance of the approach. For example, a facial recognition algorithm may perform poorly when image resolution is low or images are taken in low lighting. Or a speech-to-text system might not be used reliably to provide closed captions for online lectures because it fails to handle technical jargon.
        \item The authors should discuss the computational efficiency of the proposed algorithms and how they scale with dataset size.
        \item If applicable, the authors should discuss possible limitations of their approach to address problems of privacy and fairness.
        \item While the authors might fear that complete honesty about limitations might be used by reviewers as grounds for rejection, a worse outcome might be that reviewers discover limitations that aren't acknowledged in the paper. The authors should use their best judgment and recognize that individual actions in favor of transparency play an important role in developing norms that preserve the integrity of the community. Reviewers will be specifically instructed to not penalize honesty concerning limitations.
    \end{itemize}

\item {\bf Theory assumptions and proofs}
    \item[] Question: For each theoretical result, does the paper provide the full set of assumptions and a complete (and correct) proof?
    \item[] Answer: \answerYes{} 
    \item[] Justification: Please refer to Apps.~\ref{app:thm_proof}, \ref{app:prop_proof} and \ref{app:derive_HZ} for details.
    \item[] Guidelines:
    \begin{itemize}
        \item The answer \answerNA{} means that the paper does not include theoretical results. 
        \item All the theorems, formulas, and proofs in the paper should be numbered and cross-referenced.
        \item All assumptions should be clearly stated or referenced in the statement of any theorems.
        \item The proofs can either appear in the main paper or the supplemental material, but if they appear in the supplemental material, the authors are encouraged to provide a short proof sketch to provide intuition. 
        \item Inversely, any informal proof provided in the core of the paper should be complemented by formal proofs provided in appendix or supplemental material.
        \item Theorems and Lemmas that the proof relies upon should be properly referenced. 
    \end{itemize}

    \item {\bf Experimental result reproducibility}
    \item[] Question: Does the paper fully disclose all the information needed to reproduce the main experimental results of the paper to the extent that it affects the main claims and/or conclusions of the paper (regardless of whether the code and data are provided or not)?
    \item[] Answer: \answerYes{} 
    \item[] Justification: Please refer to Sec.~\ref{sec:exp_settings} and Apps.~\ref{app:implementation_details}, \ref{app:pseudocode} for details.
    \item[] Guidelines:
    \begin{itemize}
        \item The answer \answerNA{} means that the paper does not include experiments.
        \item If the paper includes experiments, a \answerNo{} answer to this question will not be perceived well by the reviewers: Making the paper reproducible is important, regardless of whether the code and data are provided or not.
        \item If the contribution is a dataset and\slash or model, the authors should describe the steps taken to make their results reproducible or verifiable. 
        \item Depending on the contribution, reproducibility can be accomplished in various ways. For example, if the contribution is a novel architecture, describing the architecture fully might suffice, or if the contribution is a specific model and empirical evaluation, it may be necessary to either make it possible for others to replicate the model with the same dataset, or provide access to the model. In general. releasing code and data is often one good way to accomplish this, but reproducibility can also be provided via detailed instructions for how to replicate the results, access to a hosted model (e.g., in the case of a large language model), releasing of a model checkpoint, or other means that are appropriate to the research performed.
        \item While NeurIPS does not require releasing code, the conference does require all submissions to provide some reasonable avenue for reproducibility, which may depend on the nature of the contribution. For example
        \begin{enumerate}
            \item If the contribution is primarily a new algorithm, the paper should make it clear how to reproduce that algorithm.
            \item If the contribution is primarily a new model architecture, the paper should describe the architecture clearly and fully.
            \item If the contribution is a new model (e.g., a large language model), then there should either be a way to access this model for reproducing the results or a way to reproduce the model (e.g., with an open-source dataset or instructions for how to construct the dataset).
            \item We recognize that reproducibility may be tricky in some cases, in which case authors are welcome to describe the particular way they provide for reproducibility. In the case of closed-source models, it may be that access to the model is limited in some way (e.g., to registered users), but it should be possible for other researchers to have some path to reproducing or verifying the results.
        \end{enumerate}
    \end{itemize}

\item {\bf Open access to data and code}
    \item[] Question: Does the paper provide open access to the data and code, with sufficient instructions to faithfully reproduce the main experimental results, as described in supplemental material?
    \item[] Answer: \answerYes{} 
    \item[] Justification: We use public data which can be downloaded online. The source code will be released upon acceptance.
    \item[] Guidelines:
    \begin{itemize}
        \item The answer \answerNA{} means that paper does not include experiments requiring code.
        \item Please see the NeurIPS code and data submission guidelines (\url{https://neurips.cc/public/guides/CodeSubmissionPolicy}) for more details.
        \item While we encourage the release of code and data, we understand that this might not be possible, so \answerNo{} is an acceptable answer. Papers cannot be rejected simply for not including code, unless this is central to the contribution (e.g., for a new open-source benchmark).
        \item The instructions should contain the exact command and environment needed to run to reproduce the results. See the NeurIPS code and data submission guidelines (\url{https://neurips.cc/public/guides/CodeSubmissionPolicy}) for more details.
        \item The authors should provide instructions on data access and preparation, including how to access the raw data, preprocessed data, intermediate data, and generated data, etc.
        \item The authors should provide scripts to reproduce all experimental results for the new proposed method and baselines. If only a subset of experiments are reproducible, they should state which ones are omitted from the script and why.
        \item At submission time, to preserve anonymity, the authors should release anonymized versions (if applicable).
        \item Providing as much information as possible in supplemental material (appended to the paper) is recommended, but including URLs to data and code is permitted.
    \end{itemize}

\item {\bf Experimental setting/details}
    \item[] Question: Does the paper specify all the training and test details (e.g., data splits, hyperparameters, how they were chosen, type of optimizer) necessary to understand the results?
    \item[] Answer: \answerYes{} 
    \item[] Justification: Please refer to Sec.~\ref{sec:exp_settings} and App.~\ref{app:implementation_details}.
    \item[] Guidelines:
    \begin{itemize}
        \item The answer \answerNA{} means that the paper does not include experiments.
        \item The experimental setting should be presented in the core of the paper to a level of detail that is necessary to appreciate the results and make sense of them.
        \item The full details can be provided either with the code, in appendix, or as supplemental material.
    \end{itemize}

\item {\bf Experiment statistical significance}
    \item[] Question: Does the paper report error bars suitably and correctly defined or other appropriate information about the statistical significance of the experiments?
    \item[] Answer: \answerYes{} 
    \item[] Justification: We report error bars for the main results and additional experiments
    \item[] Guidelines:
    \begin{itemize}
        \item The answer \answerNA{} means that the paper does not include experiments.
        \item The authors should answer \answerYes{} if the results are accompanied by error bars, confidence intervals, or statistical significance tests, at least for the experiments that support the main claims of the paper.
        \item The factors of variability that the error bars are capturing should be clearly stated (for example, train/test split, initialization, random drawing of some parameter, or overall run with given experimental conditions).
        \item The method for calculating the error bars should be explained (closed form formula, call to a library function, bootstrap, etc.)
        \item The assumptions made should be given (e.g., Normally distributed errors).
        \item It should be clear whether the error bar is the standard deviation or the standard error of the mean.
        \item It is OK to report 1-sigma error bars, but one should state it. The authors should preferably report a 2-sigma error bar than state that they have a 96\% CI, if the hypothesis of Normality of errors is not verified.
        \item For asymmetric distributions, the authors should be careful not to show in tables or figures symmetric error bars that would yield results that are out of range (e.g., negative error rates).
        \item If error bars are reported in tables or plots, the authors should explain in the text how they were calculated and reference the corresponding figures or tables in the text.
    \end{itemize}

\item {\bf Experiments compute resources}
    \item[] Question: For each experiment, does the paper provide sufficient information on the computer resources (type of compute workers, memory, time of execution) needed to reproduce the experiments?
    \item[] Answer: \answerYes{} 
    \item[] Justification: We provide computation resource in Apps.~\ref{app:implementation_details}.
    \item[] Guidelines:
    \begin{itemize}
        \item The answer \answerNA{} means that the paper does not include experiments.
        \item The paper should indicate the type of compute workers CPU or GPU, internal cluster, or cloud provider, including relevant memory and storage.
        \item The paper should provide the amount of compute required for each of the individual experimental runs as well as estimate the total compute. 
        \item The paper should disclose whether the full research project required more compute than the experiments reported in the paper (e.g., preliminary or failed experiments that didn't make it into the paper). 
    \end{itemize}
    
\item {\bf Code of ethics}
    \item[] Question: Does the research conducted in the paper conform, in every respect, with the NeurIPS Code of Ethics \url{https://neurips.cc/public/EthicsGuidelines}?
    \item[] Answer: \answerYes{} 
    \item[] Justification: We have review the Code of Ethics and follow it strictly.
    \item[] Guidelines:
    \begin{itemize}
        \item The answer \answerNA{} means that the authors have not reviewed the NeurIPS Code of Ethics.
        \item If the authors answer \answerNo, they should explain the special circumstances that require a deviation from the Code of Ethics.
        \item The authors should make sure to preserve anonymity (e.g., if there is a special consideration due to laws or regulations in their jurisdiction).
    \end{itemize}

\item {\bf Broader impacts}
    \item[] Question: Does the paper discuss both potential positive societal impacts and negative societal impacts of the work performed?
    \item[] Answer: \answerYes{} 
    \item[] Justification: Please refer to App.~\ref{app:impact}
    \item[] Guidelines:
    \begin{itemize}
        \item The answer \answerNA{} means that there is no societal impact of the work performed.
        \item If the authors answer \answerNA{} or \answerNo, they should explain why their work has no societal impact or why the paper does not address societal impact.
        \item Examples of negative societal impacts include potential malicious or unintended uses (e.g., disinformation, generating fake profiles, surveillance), fairness considerations (e.g., deployment of technologies that could make decisions that unfairly impact specific groups), privacy considerations, and security considerations.
        \item The conference expects that many papers will be foundational research and not tied to particular applications, let alone deployments. However, if there is a direct path to any negative applications, the authors should point it out. For example, it is legitimate to point out that an improvement in the quality of generative models could be used to generate Deepfakes for disinformation. On the other hand, it is not needed to point out that a generic algorithm for optimizing neural networks could enable people to train models that generate Deepfakes faster.
        \item The authors should consider possible harms that could arise when the technology is being used as intended and functioning correctly, harms that could arise when the technology is being used as intended but gives incorrect results, and harms following from (intentional or unintentional) misuse of the technology.
        \item If there are negative societal impacts, the authors could also discuss possible mitigation strategies (e.g., gated release of models, providing defenses in addition to attacks, mechanisms for monitoring misuse, mechanisms to monitor how a system learns from feedback over time, improving the efficiency and accessibility of ML).
    \end{itemize}
    
\item {\bf Safeguards}
    \item[] Question: Does the paper describe safeguards that have been put in place for responsible release of data or models that have a high risk for misuse (e.g., pre-trained language models, image generators, or scraped datasets)?
    \item[] Answer: \answerNA{} 
    \item[] Justification: The paper poses no such risks.
    \item[] Guidelines:
    \begin{itemize}
        \item The answer \answerNA{} means that the paper poses no such risks.
        \item Released models that have a high risk for misuse or dual-use should be released with necessary safeguards to allow for controlled use of the model, for example by requiring that users adhere to usage guidelines or restrictions to access the model or implementing safety filters. 
        \item Datasets that have been scraped from the Internet could pose safety risks. The authors should describe how they avoided releasing unsafe images.
        \item We recognize that providing effective safeguards is challenging, and many papers do not require this, but we encourage authors to take this into account and make a best faith effort.
    \end{itemize}

\item {\bf Licenses for existing assets}
    \item[] Question: Are the creators or original owners of assets (e.g., code, data, models), used in the paper, properly credited and are the license and terms of use explicitly mentioned and properly respected?
    \item[] Answer: \answerYes{} 
    \item[] Justification: The datasets and pretrained models used in this work are cited properly.
    \item[] Guidelines:
    \begin{itemize}
        \item The answer \answerNA{} means that the paper does not use existing assets.
        \item The authors should cite the original paper that produced the code package or dataset.
        \item The authors should state which version of the asset is used and, if possible, include a URL.
        \item The name of the license (e.g., CC-BY 4.0) should be included for each asset.
        \item For scraped data from a particular source (e.g., website), the copyright and terms of service of that source should be provided.
        \item If assets are released, the license, copyright information, and terms of use in the package should be provided. For popular datasets, \url{paperswithcode.com/datasets} has curated licenses for some datasets. Their licensing guide can help determine the license of a dataset.
        \item For existing datasets that are re-packaged, both the original license and the license of the derived asset (if it has changed) should be provided.
        \item If this information is not available online, the authors are encouraged to reach out to the asset's creators.
    \end{itemize}

\item {\bf New assets}
    \item[] Question: Are new assets introduced in the paper well documented and is the documentation provided alongside the assets?
    \item[] Answer: \answerNA{} 
    \item[] Justification: The paper does not release new assets.
    \item[] Guidelines:
    \begin{itemize}
        \item The answer \answerNA{} means that the paper does not release new assets.
        \item Researchers should communicate the details of the dataset\slash code\slash model as part of their submissions via structured templates. This includes details about training, license, limitations, etc. 
        \item The paper should discuss whether and how consent was obtained from people whose asset is used.
        \item At submission time, remember to anonymize your assets (if applicable). You can either create an anonymized URL or include an anonymized zip file.
    \end{itemize}

\item {\bf Crowdsourcing and research with human subjects}
    \item[] Question: For crowdsourcing experiments and research with human subjects, does the paper include the full text of instructions given to participants and screenshots, if applicable, as well as details about compensation (if any)? 
    \item[] Answer: \answerNA{} 
    \item[] Justification: The paper does not involve crowdsourcing nor research with human subjects.
    \item[] Guidelines:
    \begin{itemize}
        \item The answer \answerNA{} means that the paper does not involve crowdsourcing nor research with human subjects.
        \item Including this information in the supplemental material is fine, but if the main contribution of the paper involves human subjects, then as much detail as possible should be included in the main paper. 
        \item According to the NeurIPS Code of Ethics, workers involved in data collection, curation, or other labor should be paid at least the minimum wage in the country of the data collector. 
    \end{itemize}

\item {\bf Institutional review board (IRB) approvals or equivalent for research with human subjects}
    \item[] Question: Does the paper describe potential risks incurred by study participants, whether such risks were disclosed to the subjects, and whether Institutional Review Board (IRB) approvals (or an equivalent approval/review based on the requirements of your country or institution) were obtained?
    \item[] Answer: \answerNA{} 
    \item[] Justification: The paper does not involve crowsourcing nor research with human subjects.
    \item[] Guidelines:
    \begin{itemize}
        \item The answer \answerNA{} means that the paper does not involve crowdsourcing nor research with human subjects.
        \item Depending on the country in which research is conducted, IRB approval (or equivalent) may be required for any human subjects research. If you obtained IRB approval, you should clearly state this in the paper. 
        \item We recognize that the procedures for this may vary significantly between institutions and locations, and we expect authors to adhere to the NeurIPS Code of Ethics and the guidelines for their institution. 
        \item For initial submissions, do not include any information that would break anonymity (if applicable), such as the institution conducting the review.
    \end{itemize}

\item {\bf Declaration of LLM usage}
    \item[] Question: Does the paper describe the usage of LLMs if it is an important, original, or non-standard component of the core methods in this research? Note that if the LLM is used only for writing, editing, or formatting purposes and does \emph{not} impact the core methodology, scientific rigor, or originality of the research, declaration is not required.
    \item[] Answer: \answerNA{} 
    \item[] Justification: LLMs, if used, were only used for language editing and formatting and did not affect the method, experiments, analysis, or scientific claims.
    \item[] Guidelines:
    \begin{itemize}
        \item The answer \answerNA{} means that the core method development in this research does not involve LLMs as any important, original, or non-standard components.
        \item Please refer to our LLM policy in the NeurIPS handbook for what should or should not be described.
    \end{itemize}

\end{enumerate}

\end{document}